\documentclass[11pt, a4paper]{article}
\usepackage[letterpaper, left=1in, right=1in, top=0.9in, bottom=0.9in]{geometry}

\usepackage{amsmath}
\usepackage{amssymb}
\usepackage{amsfonts}
\usepackage{amsthm} 
\usepackage{graphicx}
\usepackage{comment}

\usepackage{thm-restate}

\usepackage{hyperref}
\usepackage{here}

\usepackage{cite}

\newtheorem{theorem}{Theorem}[section]
\newtheorem{lemma}[theorem]{Lemma}

\newtheorem{corollary}[theorem]{Corollary}

\newtheorem{observation}[theorem]{Observation}
\newtheorem{definition}[theorem]{Definition}
\newtheorem{remark}[theorem]{Remark}

\usepackage[linesnumbered,noend,ruled,vlined]{algorithm2e}

\SetKwInput{KwInput}{Input}     
\SetKwInput{KwOutput}{Output}   
\SetKwInput{KwQuestion}{Question}   
\SetKw{Kwto}{to}

\usepackage{algorithmic}

\newcommand{\algorithmiccontinue}{\textbf{continue}}
\newcommand{\Continue}{\algorithmiccontinue}

\usepackage{tabularx}

\usepackage{url}

\usepackage{makecell}
\usepackage{braket}

\usepackage{color}
\ifdefined\ShowComment
\newcommand{\terao}[1]{{\bf \color{cyan} TERAO COMMENT: #1}}
\else
\newcommand{\terao}[1]{}
\fi

 \newcommand{\MoriModifyRedColor}{true}

\ifdefined\MoriModifyRedColor
\newcommand{\ryuhei}[1]{\textcolor{red}{#1}}
\newcommand{\ryuheidel}[1]{\ryuhei{\sout{{#1}}}}

\else
\newcommand{\ryuhei}[1]{#1}
\newcommand{\ryuheidel}[1]{}

\fi

\ifdefined\TeraoModifyBlueColor
\newcommand{\tatsuya}[1]{\textcolor{blue}{#1}}
\newcommand{\tatsuyadel}[1]{\tatsuya{\sout{{#1}}}}

\newcommand{\tatsuyawithdraw}[1]{\tatsuya{\sout{{#1}}}}
\else
\newcommand{\tatsuya}[1]{#1}
\newcommand{\tatsuyadel}[1]{}

\newcommand{\tatsuyawithdraw}[1]{}
\fi

\global\long\def\T{\mathcal{T}}%

\title{Parameterized Quantum Query Algorithms for Graph Problems\thanks{
The authors thank Yusuke Kobayashi and Kazuhisa Makino for their insightful comments.
This work was partially supported by JST FOREST Program Grant Number JPMJFR216V and JSPS KAKENHI Grant Numbers JP20H04138, JP20H05966, JP22H00522, and JP24KJ1494.
}}
\author{Tatsuya Terao\thanks{
Research Institute for Mathematical Sciences, Kyoto University.
E-mail: ttatsuya@kurims.kyoto-u.ac.jp
}
\and Ryuhei Mori\thanks{
Graduate School of Mathematics, Nagoya University.
E-mail: mori@math.nagoya-u.ac.jp
}
}

\date{}

\begin{document}

\maketitle

\begin{abstract}
In this paper, we consider the parameterized quantum query complexity for graph problems. We design parameterized quantum query algorithms for {\sc $k$-vertex cover} and {\sc $k$-matching} problems, and present lower bounds on the parameterized quantum query complexity. Then, we show that our quantum query algorithms are optimal up to a constant factor when the parameters are small. Our main results are as follows.

\textbf{Parameterized quantum query complexity of vertex cover.}
In the {\sc $k$-vertex cover} problem, we are given an undirected graph $G$ with $n$ vertices and an integer $k$, and the objective is to determine whether $G$ has a vertex cover of size at most $k$. We show that the quantum query complexity of the {\sc $k$-vertex cover} problem is $O(\sqrt{k}n + k^{3/2}\sqrt{n})$ in the adjacency matrix model. For the design of the quantum query algorithm, we use the method of kernelization, a well-known tool for the design of parameterized classical algorithms, combined with Grover's search.

\textbf{Parameterized quantum query complexity of matching.}
In the {\sc $k$-matching} problem, we are given an undirected graph $G$ with $n$ vertices and an integer $k$, and the objective is to determine whether $G$ has a matching of size at least $k$. We show that the quantum query complexity of the {\sc $k$-matching} problem is $O(\sqrt{k}n + k^{2})$ in the adjacency matrix model. We obtain this upper bound by using Grover's search carefully and analyzing the number of Grover's searches by making use of potential functions. We also show that the quantum query complexity of the maximum matching problem is $O(\sqrt{p}n + p^{2})$ where $p$ is the size of the maximum matching. For small $p$, it improves known bounds $\tilde{O}(n^{3/2})$ for bipartite graphs [Blikstad--v.d.Brand--Efron--Mukhopadhyay--Nanongkai, FOCS 2022] and $O(n^{7/4})$ for general graphs [Kimmel--Witter, WADS 2021].

\textbf{Lower bounds on parameterized quantum query complexity.}
We also present lower bounds on the quantum query complexities of the {\sc $k$-vertex cover} and {\sc $k$-matching} problems. The lower bounds prove the optimality of the above parameterized quantum query algorithms up to a constant factor when $k$ is small. Indeed, the quantum query complexities of the {\sc $k$-vertex cover} and {\sc $k$-matching} problems are both $\Theta(\sqrt{k} n)$ when $k = O(\sqrt{n})$ and $k = O(n^{2/3})$, respectively.

\end{abstract}

\newpage

\section{Introduction} \label{sec:intro}

The query complexity for Boolean functions has been studied much both on the classical and quantum settings.
In contrast to the circuit complexity, the quantum advantages on the query complexity have been rigorously proved for several problems~\cite{grover1996fast,beals2001quantum}.
While the classical query complexities are trivial for many graph problems, the quantum query complexities are non-trivially small for many graph problems~\cite{berzina2004quantum,durr2006quantum,kimmel2021query,beigi2020quantum,childs2012quantum,belovs2012span,jeffery2023quantum,cade2018time,zhang2004power,blikstad2022nearly,carette2020extended,le2014improved,lee2013improved,jeffery2013nested,magniez2007quantum,belovs2012span2,le2017quantum,ambainis2008quantumISAAC,sun2004graph,dorn2005quantum,childs2003quantum,lee2012learning,anderson2023improved,ambainis2011quantum,jarret2018quantum,apers2021quantum,aricnvs2015span,dorn2009quantum,apers2022quantum,zhu2012quantum}.

In this paper, we consider the \emph{parameterized} quantum query complexity for graph problems.
In terms of the classical time complexity, the theory of parameterized time complexity classifies NP-hard problems on a finer scale.
Some graph problems can be solved in polynomial time, e.g., the maximum matching problem.
On the other hand, several graph problems are NP-hard, e.g., the minimum vertex cover problem, which can not be solved in polynomial time unless $\rm{P}=\rm{NP}$.
For the detailed analysis of the NP-hard problems, we have studied the parameterized complexity, where we design and analyze algorithms for problems with fixed parameters. 
Then, our goal is to design algorithms that perform efficiently when the parameters are small.
Formally, the definition of a \emph{parameterized problem} is a language $L \subseteq \Sigma^* \times \mathbb{N}$, where $\Sigma$ is a fixed finite alphabet.
For an instance $(x, k) \in \Sigma^* \times \mathbb{N}$, we call $k$ the \emph{parameter}.
The main focus in the field of the parameterized complexity is whether there exists an algorithm to solve the parameterized problem in $f(k) \cdot |x|^{O(1)}$ time, where $f$ is an arbitrary computable function and $|x|$ is the size of $x$.
A parameterized problem is called \emph{Fixed Parameter Tractable (FPT)} if such an algorithm exists.
The running time of the FPT algorithms is polynomial in the input size $|x|$ with a constant exponent independent of the parameter $k$ if the parameter $k$ is constant.
This parameter is typically taken to be the size of the solution.
Then, we consider the problem of deciding whether there is a solution of size bounded by $k$.
For example, the $k$-vertex cover problem, which asks whether a given graph has a vertex cover of size at most $k$, can be solved in $O(2^k \cdot n)$ time.
This means that the $k$-vertex cover problem is FPT.

For the adjacency matrix model, the quantum query complexity for graph problems is trivially $O(n^2)$, where $n$ is the number of vertices.
For the parameterized graph problems, our objective is to design quantum query algorithms with the query complexity $f(k) \cdot n^c$, where $c$ is some constant smaller than $2$.
We will introduce a class of quantum query algorithms called \emph{Fixed Parameter Improved (FPI)} algorithms.
For the definition of FPI, we consider an unparameterized problem of the parameterized problem.
For example, the unparameterized problems of the $k$-vertex cover and $k$-matching problems are the minimum vertex cover and the maximum matching problems, respectively.
If the query complexity of the unparameterized problem is $\Omega(n^e)$, the parameterized problem is FPI if its query complexity is at most $f(k) \cdot n^c$ for some constant $c<e$ independent of $k$.
When we can take $f(k)$ as a polynomial, we say that the parameterized problem is polynomial FPI.

For the minimum vertex cover problem and the maximum matching problem, the quantum query complexity is $\Omega(n^{3/2})$~\cite{zhang2004power}.
In this work, we prove that the quantum query complexities of the $k$-vertex cover problem and $k$-matching problem are $O(\sqrt{k}n+k^{3/2}\sqrt{n})$ and $O(\sqrt{k}n+k^2)$, respectively.
Hence, the $k$-vertex cover and $k$-matching problems are polynomial FPI.

We also prove lower bounds $\Omega(\sqrt{k}n)$ of the quantum query complexities of the $k$-vertex cover and $k$-matching problems using the adversary method.
Hence, the quantum query algorithms in this paper are optimal up to a constant factor for small $k$.

Note that no non-trivial quantum query complexity $O(n^{2-\epsilon})$ has been obtained for the minimum vertex cover problem.
The best known upper bound on the quantum query complexity of the maximum matching problem is $O(n^{7/4})$~\cite{kimmel2021query}.

In summary, this work considers the parameterized quantum query complexity for graph problems, introduces FPI algorithms, and presents quantum query algorithms for the $k$-vertex cover and $k$-matching problems that are polynomial FPI and optimal for small $k$.

\subsection{Our Contribution} \label{subsec:our_contribution}

\subsubsection{Parameterized Quantum Query Complexity of Vertex Cover.}
In this paper, we consider the quantum query complexity of the {\sc $k$-vertex cover} problem,
which is one of the central problems in the study of parameterized computation.
This problem is also one of Karp's 21 NP-complete problems \cite{karp2010reducibility}.
A vertex set $S \subseteq V(G)$ is a \emph{vertex cover} if every edge of the graph has at least one endpoint in $S$.
In the {\sc $k$-vertex cover} problem, we are given an undirected graph $G = (V, E)$ and an integer $k$, and 
the objective is to determine whether $G$ has a vertex cover of size at most $k$.

Here, we mention some known results related to our work.
There have been many studies to improve the running time of an FPT algorithm for the {\sc $k$-vertex cover}
in the literature
\cite{chen2010improved,balasubramanian1998improved,buss1993nondeterminism,chandran2004refined,chen2001vertex,downey1992fixed,niedermeier1999upper,niedermeier2003efficient,stege1999improved,harris_et_al:LIPIcs.STACS.2024.40}.
The state-of-the-art algorithm, given by Harris--Narayanaswamy \cite{harris_et_al:LIPIcs.STACS.2024.40}, has the time complexity $O(1.2575^k \cdot \text{poly}(n))$.

The minimum vertex cover problem is the optimization version of the {\sc $k$-vertex cover} problem.
In this problem, the objective is to find the smallest vertex cover of $G$. 
The trivial upper bound $O(n^2)$ on the quantum query complexity of the minimum vertex cover problem in the adjacency matrix model has not been improved.
The largest known lower bound on the quantum query complexity of the minimum vertex cover problem is $\Omega(n^{3/2})$ in the adjacency matrix model \cite{zhang2004power}.\footnote{Zhang \cite{zhang2004power} showed that the quantum query complexity of the maximum matching problem is $\Omega(n^{3/2})$. For bipartite graphs, the minimum vertex cover has the same cardinality as the maximum matching by well-known K\"{o}nig's theorem (see e.g., \cite[Theorem 16.2]{schrijver2003combinatorial}). Thus, we can easily obtain an $\Omega(n^{3/2})$ lower bound for the minimum vertex cover problem.}

The first main contributions of this paper is to obtain an upper bound on the quantum query complexity of the {\sc $k$-vertex cover} problem.
\begin{theorem} \label{thm:main_theorem_matrix_model}
The quantum query complexity for finding a vertex cover of size at most $k$ or determining that there does not exist a vertex cover of size at most $k$ with bounded error is $O(\sqrt{k}n + k^{3/2}\sqrt{n})$ in the adjacency matrix model. 
\end{theorem}

We also provide a lower bound on the quantum query complexity of the {\sc$k$-vertex cover} problem.

\begin{restatable}{theorem}{quantumvclb} \label{thm:main_theorem_matrix_model_lower_bound}
For any constant $\epsilon > 0$, given a graph $G$ with $n$ vertices in the adjacency matrix model and an integer $k \leq (1 - \epsilon) n$,
the quantum query complexity of deciding whether $G$ has a vertex cover of size at most $k$ with bounded error is $\Omega(\sqrt{k} n)$.
\end{restatable}

The upper and lower bounds on the quantum query complexity of the {\sc $k$-vertex cover} problem are summarized in Figure \ref{fig:vc_complexity}.
From Theorems \ref{thm:main_theorem_matrix_model} and \ref{thm:main_theorem_matrix_model_lower_bound}, we obtain the following corollary.

\begin{corollary} \label{cor:optimal_vertex_cover}
The quantum query complexity of deciding whether $G$ has a vertex cover of size at most $k$ with bounded error is $\Theta(\sqrt{k}n)$ when $k = O(\sqrt{n})$ in the adjacency matrix model.
\end{corollary}

\begin{figure}[t]
\centering
\includegraphics[width=10cm]{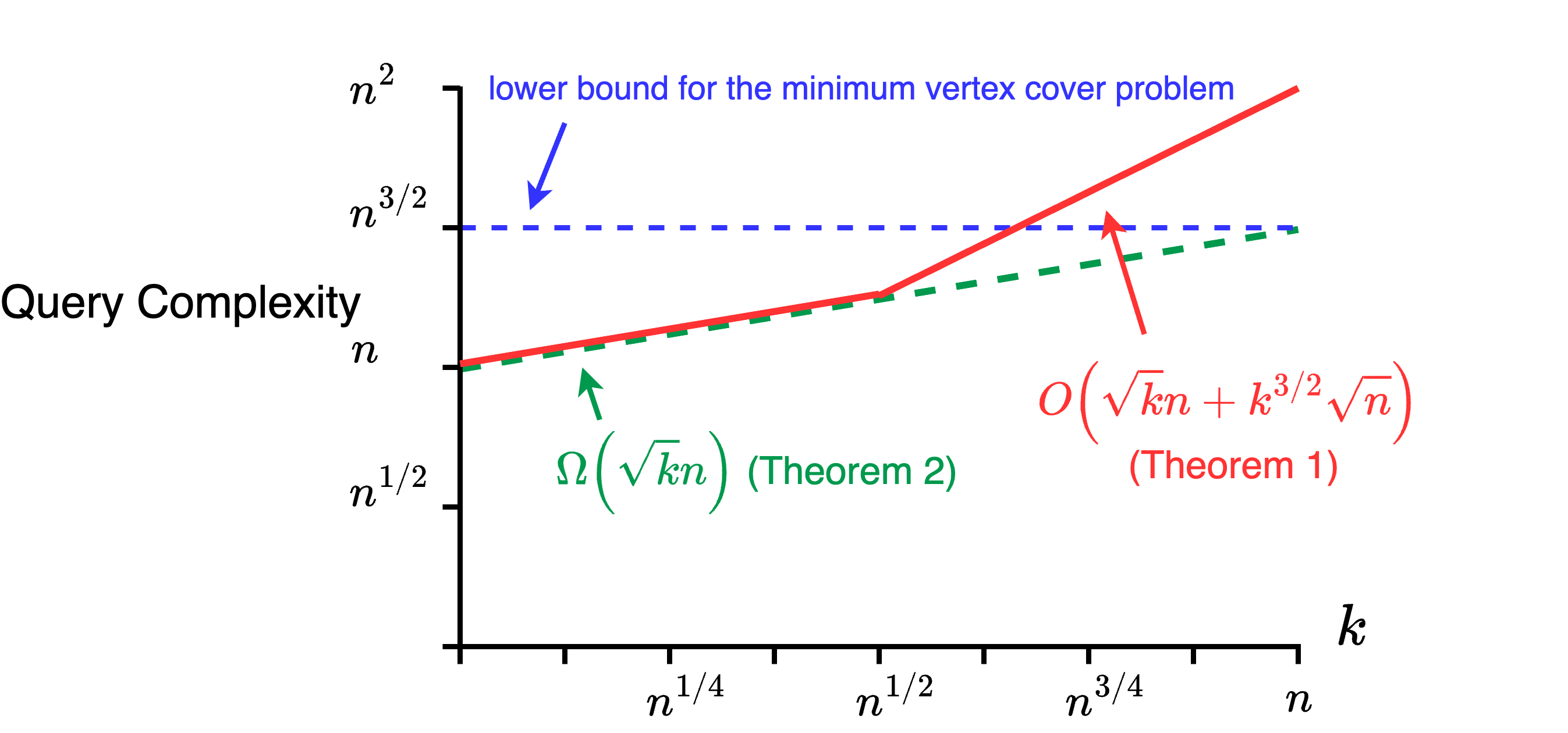}
\caption{The quantum query complexity of the {\sc $k$-vertex cover} problem.}
\label{fig:vc_complexity}
\end{figure}

A graph $G$ has a vertex cover of size at most $k$ if and only if the complement graph $\overline{G}$ has a clique of size at least $n-k$.
Hence, our result given in Corollary \ref{cor:optimal_vertex_cover} implies new bounds on the quantum query complexity of the $k$-clique problem, which is a well-studied problem in the quantum query model \cite{childs2003quantum, magniez2007quantum, zhu2012quantum, lee2012learning, le2014improved,lee2013improved,jeffery2013nested,belovs2012span2, carette2020extended, le2017quantum}.
In the $k$-clique problem, the objective is to find a clique of size at least $k$.
The best known upper bound on the quantum query complexity of the $k$-clique problem in the adjacency matrix model is $O(n^{2 - 2/k - g(k)})$ where $g(k) = O(1/k^3)$ is a strictly positive function~\cite{zhu2012quantum, lee2012learning}.
Note that only the trivial $\Omega(n)$ lower bound is known.
While non-trivial optimal query complexity of the $k$-clique problem has not been known for any $k$, Corollary \ref{cor:optimal_vertex_cover} establishes non-trivial optimal query complexity $\Theta(\sqrt{k}n)$ of detecting cliques of size at least $n - k$ for any $k=O(\sqrt{n})$.

We also provide a lower bound on the randomized query complexity of the {\sc $k$-vertex cover} problem.

\begin{restatable}{proposition}{classicalvclb} \label{thm:main_theorem_classical_lower_bound}
Given a graph $G$ with n vertices in the adjacency matrix model and an integer $k < n - 1$,\footnote{We consider only the setting where $k < n - 1$, since an input instance is always a yes-instance in the setting where $k \geq n - 1$.}
the randomized query complexity of deciding whether $G$ has a vertex cover of size at most $k$ with bounded error is $\Omega(n^2)$.
\end{restatable}

\subsubsection{Parameterized Quantum Query Complexity of Matching.}
We also consider the quantum query complexity of the {\sc $k$-matching} problem,
which is the parameterized version of the maximum matching problem.
An edge set $M \subseteq E(G)$ is a \emph{matching} if each vertex in $V(G)$ appears in at most one edge in $M$.
In the {\sc $k$-matching} problem, we are given an undirected graph $G = (V, E)$ and an integer $k$, and the objective is to determine whether $G$ has a matching of size at least $k$.

We obtain an upper bound on the quantum query complexity of the {\sc $k$-matching} problem.
\begin{theorem} \label{thm:main_theorem_parameterized_matching_matrix_model}
The quantum query complexity of finding a matching of size at least $k$ or determining that there does not exist a matching of size at least $k$ with bounded error is $O(\sqrt{k}n + k^{2})$ in the adjacency matrix model. 
\end{theorem}

We also provide a lower bound on the quantum query complexity of the {\sc$k$-matching} problem.

\begin{restatable}{theorem}{quantummatchinglb} \label{thm:main_theorem_matrix_model_lower_bound_matching}
Given a graph $G$ with $n$ vertices in the adjacency matrix model and an integer $k \leq n / 2$,
the quantum query complexity of deciding whether $G$ has a matching of size at least $k$ with bounded error is $\Omega(\sqrt{k} n)$.
\end{restatable}

Zhang \cite{zhang2004power} showed that the quantum query complexity of deciding whether $G$ has a perfect matching is $\Omega(n^{3/2})$.
Then, Theorem \ref{thm:main_theorem_matrix_model_lower_bound_matching} is a generalization of the result of Zhang.

The upper and lower bounds on the quantum query complexity of the {\sc $k$-matching} problem are summarized in Figure \ref{fig:matching_complexity}.
From Theorems \ref{thm:main_theorem_parameterized_matching_matrix_model} and \ref{thm:main_theorem_matrix_model_lower_bound_matching}, we obtain the following corollary.

\begin{corollary}
The quantum query complexity for deciding whether $G$ has a matching of size at least $k$ with bounded error is $\Theta(\sqrt{k}n)$ when $k = O(n^{2/3})$ in the adjacency matrix model.
\end{corollary}

\begin{figure}[t]
\centering
\includegraphics[width=10cm]{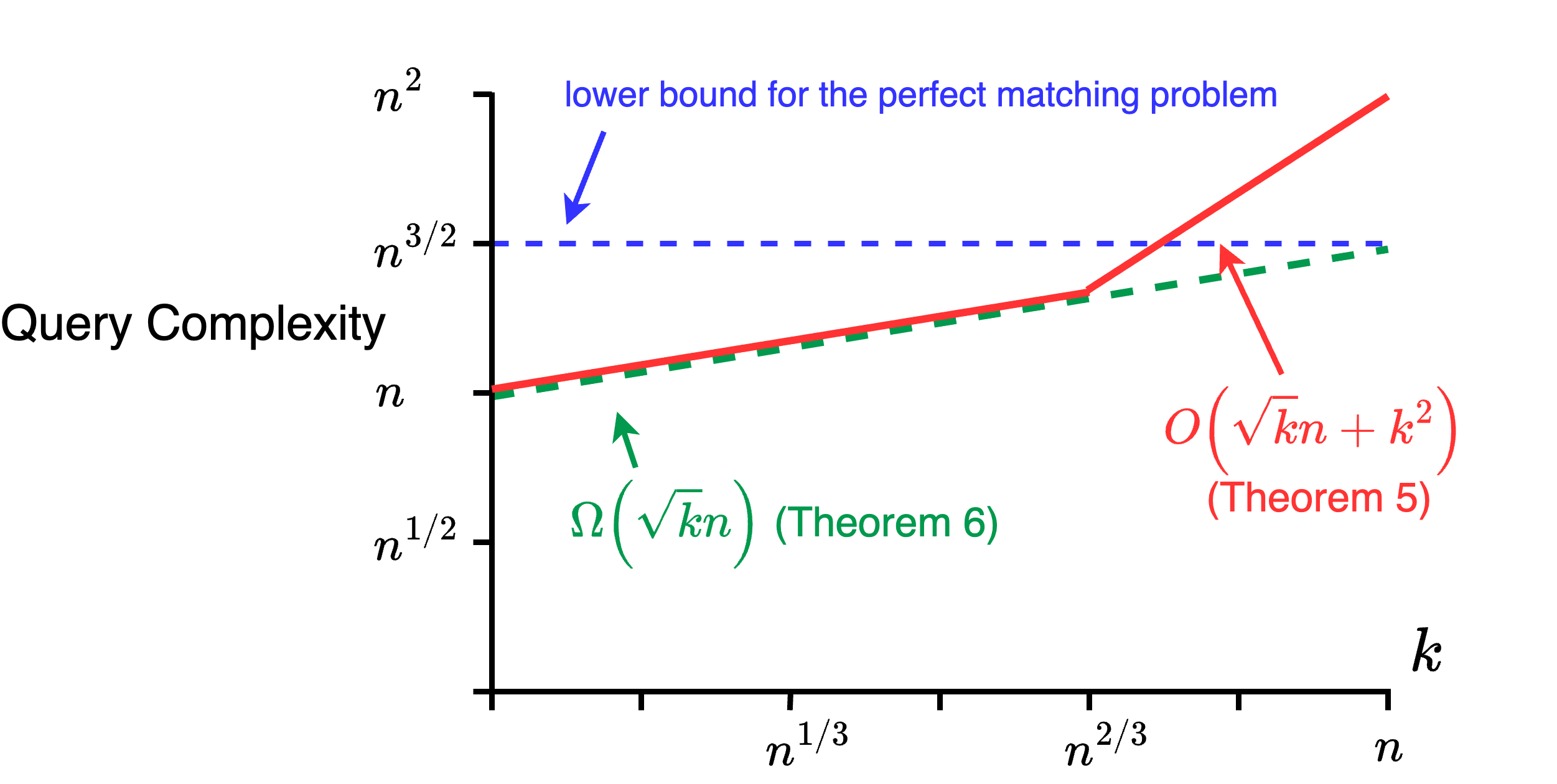}
\caption{The quantum query complexity of the {\sc $k$-matching} problem.}
\label{fig:matching_complexity}
\end{figure}

We also provide a lower bound on the randomized query complexity of the {\sc $k$-matching} problem.

\begin{restatable}{proposition}{classicalmatchinglb} \label{thm:main_theorem_matrix_model_classical_lower_bound_matching}
Given a graph $G$ with $n$ vertices in the adjacency matrix model and an integer $k \leq n / 2$,
the randomized query complexity of deciding whether $G$ has a matching of size at least $k$ with bounded error is $\Omega(n^2)$.
\end{restatable}

By using Theorem \ref{thm:main_theorem_parameterized_matching_matrix_model}, we also obtain the quantum query complexity of the maximum matching problem.
Here, let $p$ denote the size of the maximum matching.

\begin{theorem} \label{thm:main_theorem_matching_matrix_model}
The quantum query complexity for finding a maximum matching with bounded error is $O(\sqrt{p}n + p^{2})$ in the adjacency matrix model. 
\end{theorem}

This query complexity improves upon known results when $p$ is small.
Lin--Lin\cite{lin2015upper} obtained an $O(p^{1/4}n^{3/2})$ upper bound on the quantum query complexity of the maximum matching problem for bipartite graphs in the adjacency matrix model\footnote{Their paper only states that the quantum query complexity of the bipartite maximum matching problem is $O(n^{7/4})$ in the adjacency matrix model.
We see that this upper bound can be easily improved to $O(p^{1/4}n^{3/2})$.
Their quantum query algorithm relies on the classical algorithm by Hopcroft--Karp \cite{hopcroft1973n}.
In the algorithm of Hopcroft--Karp, they repeatedly find a maximal set of vertex disjoint shortest augmenting paths.
To obtain an $O(n^{7/4})$ upper bound, Lin--Lin use the fact that $O(\sqrt{n})$ iterations of these processes suffice. 
It is known that only $O(\sqrt{p})$ iterations suffice; see \cite[Theorem 3]{hopcroft1973n}.
Then, we can obtain an $O(p^{1/4}n^{3/2})$ upper bound for the bipartite maximum matching problem.
}, which is the first nontrivial upper bound.
By using the method developed by Lin--Lin, Kimmel--Witter \cite{kimmel2021query} obtained an $O(p^{1/4}n^{3/2})$ upper bound for general graphs.
Moreover, Blikstad--v.d.Brand--Efron--Mukhopadhyay--Nanongkai \cite{blikstad2022nearly} showed that the quantum query complexity of the bipartite maximum matching problem is $\tilde{O}(n^{3/2})$ in the adjacency matrix model\footnote{The $\tilde{O}$ notation hides factors polynomial in $\log n$.}.

\subsubsection{Fixed Parameter Improved Algorithms}
For a better understanding of the main results of this work, we introduce a new class of query algorithms for graph problems, which we call \emph{fixed parameter improved (FPI)} algorithms.
First, we define parameterized graph problems.
\begin{definition}[Parameterized graph problem]
    In a parameterized graph problem, inputs are the number of vertices $n$, an oracle for a $n$-vertex graph $G$, typically in the adjacency matrix model, and a nonnegative integer $k$ as a parameter.
    Here, the parameter $k$ is polynomially bounded in $n$.
    The output of the problem is some property of the graph $G$ parameterized by $k$.
\end{definition}
Then, we introduce unparameterized graph problems of parameterized graph problems.
\begin{definition}[Unparameterized graph problem]
    In an unparameterized graph problem of a parameterized graph problem $\mathcal{P}$, inputs are the number of vertices $n$, and an oracle for a $n$-vertex graph $G$.
    The output of the unparameterized problem is a list of outputs of the parameterized problem $\mathcal{P}$ for all possible parameters $k$.
\end{definition}

\begin{remark}
    In many cases, a parameterized graph problem $\mathcal{P}$ is binary, i.e., the output is YES or NO, and monotone, i.e., if the output for $(G, k)$ is YES or NO, then the output for $(G, k + 1)$ is also YES or NO, respectively.
    In this case, the unparameterized graph problem $\mathcal{Q}$ for $\mathcal{P}$ is essentially an optimization problem that asks the minimum $k$ for which $(G, k)$ is a YES or NO instance of $\mathcal{P}$.
\end{remark}

Now, we define the fixed parameter improved algorithms.

\begin{definition}[Fixed parameter improved algorithms]
    Let $\mathcal{P}$ and $\mathcal{Q}$ be a parameterized graph problem and its unparameterized graph problem, respectively.
    Let 
    \begin{align*}
        C(\mathcal{Q}) &:= \inf\left\{c\mid \text{The query complexity of $\mathcal{Q}$ is $O(n^c)$}\right\}.
    \end{align*}
    A query algorithm $\mathcal{A}$ for $\mathcal{P}$ is said to be fixed parameter improved if there exists a constant $c< C(\mathcal{Q})$ and a function $f(k)$ such that the query complexity of $\mathcal{A}$ is at most $f(k) n^c$.
    When $f(k)$ is a polynomial, $\mathcal{A}$ is said to be polynomial FPI.
\end{definition}

The concept of FPI can be regarded as a natural generalization of FPT.
The correspondence between FPI and FPT is summarized in Table~\ref{tbl:FPI}.
The quantum query algorithms in Theorems~\ref{thm:main_theorem_matrix_model} and~\ref{thm:main_theorem_parameterized_matching_matrix_model} are polynomial FPI since $C(\mathcal{Q})\ge 3/2$ for both the problems~\cite{zhang2004power}.
On the other hand, the quantum query algorithms for the $k$-clique problem in~\cite{magniez2007quantum,lee2012learning,zhu2012quantum} with query complexity $n^{2-O(1/k)}$ are not FPI.

Note that the quantum algorithms in this paper are \emph{fixed parameter linear (FPL)}, which means that the query complexity is at most $f(k)n$ for some function $f$. Since we can take $f(k)$ as a polynomial, the quantum algorithms in this paper are polynomial FPL.

\begin{table}[t]
   \centering
   \caption{The correspondence between FPI and FPT.}\label{tbl:FPI}
   \begin{tabular}{|c|c|c|c|}
       \hline
        & \makecell{Unparameterized\\ complexity} & Not FPT/FPI & FPT/FPI\\
       \hline
       Time complexity & $n^{\omega(1)}$ & $n^{\omega_k(1)}$ & $\exists c<\infty\quad f(k) n^c$\\
       \hline
       Query complexity & $O\left(n^{C(\mathcal{Q})+\epsilon}\right)\forall\epsilon>0$ & $n^{C(\mathcal{Q})-o_k(1)}$ & $\exists c<C(\mathcal{Q})\quad f(k) n^c$\\
       \hline
   \end{tabular}
\end{table}

\subsection{Overview of Our Upper Bounds}

\subsubsection{Grover's Search Algorithm}

To obtain our quantum query algorithms for the {\sc $k$-vertex cover} and {\sc $k$-matching} problems, we only use some variants of Grover's search algorithm \cite{grover1996fast}, which is a fundamental tool for quantum query algorithms.

In Grover's search algorithm, we are given oracle access to a function $f \colon [N] \rightarrow \{0, 1\}$ where $K:=|f^{-1}(1)|$ is positive.
Grover's search algorithm finds some $x \in f^{-1}(1)$ with $\Theta(\sqrt{N / K})$ expected number of queries~\cite{{boyer1998tight}}; see Lemma \ref{lem:grover_expected_find} for details.
This algorithm does not require prior knowledge of $K$.

As a simple application of Grover's search algorithm,  for given $1 \leq R \leq N$, we can compute some subset $S\subseteq f^{-1}(1)$ of size $\min\{R, K\}$ with probability at least $2/3$ with query complexity $O(\sqrt{R N})$ queries (see e.g., \cite{brassard2002quantum}, \cite[Fact 2.1]{durr2006quantum}, and \cite[Theorem 13]{apers2021quantum}).

In our quantum query algorithms, we use Grover's search algorithm to find some edges in $G$.
Note that, in the adjacency matrix model, we are given query access to a function $E_{\mathrm G} \colon \binom{V(G)}{2} \rightarrow \{0, 1\}$ where $E_{\mathrm G}(v, u) = 1$ if and only if $(v, u) \in E(G)$.

\subsubsection{Overview of Our Parameterized Quantum Query Algorithm for Vertex Cover} \label{subsec:technical_overview}

\subparagraph*{The parameterized quantum query algorithm for the \texorpdfstring{$k$}{TEXT}-vertex cover.}

It is easy to obtain a bounded error quantum algorithm for the {\sc $k$-vertex cover} problem with $O(\sqrt{k}n^{3/2})$ queries.
To obtain this algorithm, we use the following easy observation:
if there are more than $k (n - 1)$ edges in the input graph $G$, $G$ does not have a vertex cover of size at most $k$.
In the algorithm, we check whether there are at most $k (n - 1)$ edges in $G$, and if so, we find all edges in $G$.
By Grover's search algorithm, this requires $O(\sqrt{k (n - 1) \cdot n^2}) = O(\sqrt{k} n^{3/2})$ queries.
After finding all edges in $G$, we simply apply a known classical algorithm for the {\sc $k$-vertex cover} problem.
Then, the quantum query complexity of the {\sc $k$-vertex cover} problem is $O(\sqrt{k}n^{3/2})$.
However, this upper bound is far from the optimal.

\subparagraph*{Kernelization.}

For the improvement of the query complexity, we use the technique of \emph{kernelization}, which is a important tool for designing parameterized algorithms~\cite{fomin2019kernelization, guo2007invitation}.
A kernelization is an efficient preprocessing algorithm that transforms a given instance to a smaller instance called a \emph{kernel}.
Here, we formally define a kernelization for a parameterized problem $L \subseteq \Sigma^* \times \mathbb{N}$, where $\Sigma$ is a fixed finite alphabet.
A kernelization is an algorithm that maps an instance $(x, k)$ of $L$ to another instance $(x', k')$ of $L$ with running time polynomial in $|x|$ and $k$, requiring that $|x'|+k'$ is bounded by a function of $k$.
Here, we also require that $(x,k)\in L$ if and only if $(x',k')\in L$. 
Then, it is sufficient to check whether $(x',k')\in L$.
The output $(x',k')$ is called a \emph{kernel} of $(x,k)$.

Typically, for graph problems, a kernel $(G',k')$ of an instance $(G,k)$ satisfies $V(G')\subseteq V(G)$ and $E(G')\subseteq E(G)$, i.e., $G'$ is a subgraph of $G$.
In this work, we consider a quantum query algorithm that outputs a kernel 
 $(G',k')$ directly from $k$ and the given oracle for $G$.
Since the size of kernels is bounded by a function of $k$, the number of edges in $G'$ is bounded by a function of $k$.
This observation suggests that the kernelization technique may yield low-query quantum algorithms for parameterized problems.
In this paper, we demonstrate, for the first time, the utility of the kernelization technique in the context of parameterized quantum query complexity.

\subparagraph*{Quantum query algorithm for kernelization.}

We introduce a new framework, called \emph{quantum query kernelization}, to obtain better quantum query complexity; see Section \ref{subsec:kernelization} for details.
A quantum query kernelization algorithm outputs a kernel $(G', k')$ as a bit string.
In the quantum query setting, the input graph $G$ is provided as a quantum oracle, through which quantum algorithms access $G$. 
Once we obtain a kernel $(G',k')$, we do not require more queries since we can decide whether $(G',k')$ is a yes-instance by an arbitrary classical algorithm for $(G', k')$.

Our quantum query kernelization algorithm is based on a known classical kernelization algorithm.
However, not all classical efficient kernelization algorithms give efficient quantum query kernelization algorithms since they are not necessarily suitable for Grover's search algorithm.
Hence, we have to choose appropriate kernels to derive better upper bounds.

\subparagraph*{Classical kernelization algorithm for the \texorpdfstring{$k$}{TEXT}-vertex cover problem.}

In this section, we introduce the classical kernelization algorithm for the $k$-vertex cover problem which will be used for designing the quantum query algorithm in Theorem~\ref{thm:main_theorem_matrix_model}.
We first find a maximal matching $M$ of $G$, which is not a proper subset of any other matching.
If the size of $M$ is at least $k + 1$, then we conclude that the input instance $(G, k)$ is a no-instance.
Let $V(M)\subseteq V$ denote the set of all endpoints of edges in $M$.
We initially set $(G', k') = (G, k)$.
For each vertex $v \in V(M)$, if $v$ has degree at least $k + 1$, then we update $(G', k')$ to $(G' - v, k' - 1)$.\footnote{For a vertex $v \in V(G)$, let $G - v$ denote the subgraph of $G$ induced by $V(G) \setminus \{ v \}$.}
Then, we return the instance $(G', k')$.

Now, the instance $(G', k')$ obtained by the above procedure is a yes-instance if and only if $(G, k)$ is a yes-instance.
This follows from the following simple observation.
\begin{observation}[{\cite{buss1993nondeterminism}}; see also {\cite[Section 2.2.1]{cygan2015parameterized}}] \label{obs:vertex_cover_kernelization}
If $G$ has a vertex $v$ of degree at least $k + 1$, then $v$ must be in every vertex cover of size at most $k$.
\end{observation}
This is because if a vertex $v$ is not in a vertex cover, then the vertex cover must contain all neighbors of $v$ to cover all edges incident to $v$.

Furthermore, $(G', k')$ obtained by the above procedure has size depending only on $k$ and not on $n$.
Let $S\subseteq V(M)$ be a subset of $V(M)$ that are contained in $G'$.
Then, $G'$ is a subgraph of $G$ induced by $S$ and it's neighborhoods in $G$ since $M$ is a maximal matching of $G$.
Note that any $v\in S$ has degree at most $k$ in $G'$.
This fact implies that $G'$ has at most $k \cdot |S|\le k \cdot |V(M)|\le 2k^2$ edges.
Note that $G'$ consists of the maximal matching $M$ (with some missing vertices due to the deletion procedure) and a independent set $V\setminus V(M)$ as in Figure~\ref{fig:kernelization}.

A similar idea was used in a parameterized streaming algorithm for the vertex cover problem~\cite{chitnis2014parameterized}.

\subparagraph*{The optimal quantum query complexity of the \texorpdfstring{$k$}{TEXT}-vertex cover for small \texorpdfstring{$k$}{TEXT}.}

Now, we present a quantum query kernelization algorithm based on the above classical kernelization algorithm.
In the above classical kernelization algorithm, there are multiple possible choices of the maximal matching and the order of the vertex eliminations of degrees greater than $k$.
Let $\mathcal{I}^{{\rm vc}}_{(G, k)}$ denote the set of all possible instances that can be returned by the above classical kernelization algorithm for an input instance $(G, k)$.
In our quantum query kernelization algorithm, we obtain a kernel $(G', k') \in \mathcal{I}^{{\rm vc}}_{(G, k)}$ as a bit string with a small number of queries.

\begin{restatable}{theorem}{quantumkenerl} \label{prop:vc_quantum_kernel}
Given a graph $G$ with n vertices in the adjacency matrix model and an integer $k$,
there is a bounded error quantum algorithm with $O(\sqrt{k}n + k^{3/2}\sqrt{n})$ queries that outputs an instance $(G', k') \in \mathcal{I}^{{\rm vc}}_{(G, k)}$, or otherwise determines that the instance $(G, k)$ is a no-instance of the $k$-vertex cover problem.
\end{restatable}

Obviously, Theorem~\ref{prop:vc_quantum_kernel} implies Theorem~\ref{thm:main_theorem_matrix_model}.
In our quantum query kernelization algorithm, we first find a matching of size at least $k + 1$ or a maximal matching.
Grover's search algorithm yields a quantum algorithm with $O(\sqrt{k} n)$ queries; see Section \ref{sec:algorithm_for_threshold_matching} for details.
If we find a matching of size at least $k + 1$, we conclude that the input instance $(G, k)$ is a no-instance.
Assume that we obtain a maximal matching $M$ of size at most $k$.
We initialize sets $V'=V(G)$, $E'=\emptyset$ and an integer $k' = k$.
For each vertex $v \in V(M)$, we find all edges incident to $v$ if $v$ has degree at most $k$ or, otherwise determine that $v$ has degree at least $k + 1$.
If $v$ has degree at most $k$, we add all edges incident to $v$ to $E'$.
If $v$ has degree at least $k + 1$, then we remove $v$ from $V'$ and update $k'$ to $k' - 1$.
For each $v\in V(M)$, Grover's search algorithm finds at most $k+1$ edges incident to $v$ from $n-1$ candidates with $O(\sqrt{kn})$ queries.
Hence, in total, we need $O(k^{3/2} \sqrt{n})$ queries.
Then, we obtain a kernel $(G'=(V',E'), k') \in \mathcal{I}^{{\rm vc}}_{(G, k)}$.
The total number of queries in the above quantum query kernelization algorithm is $O(\sqrt{k}n + k^{3/2}\sqrt{n})$.

\begin{figure}[t]
\centering
\includegraphics[width=7cm]{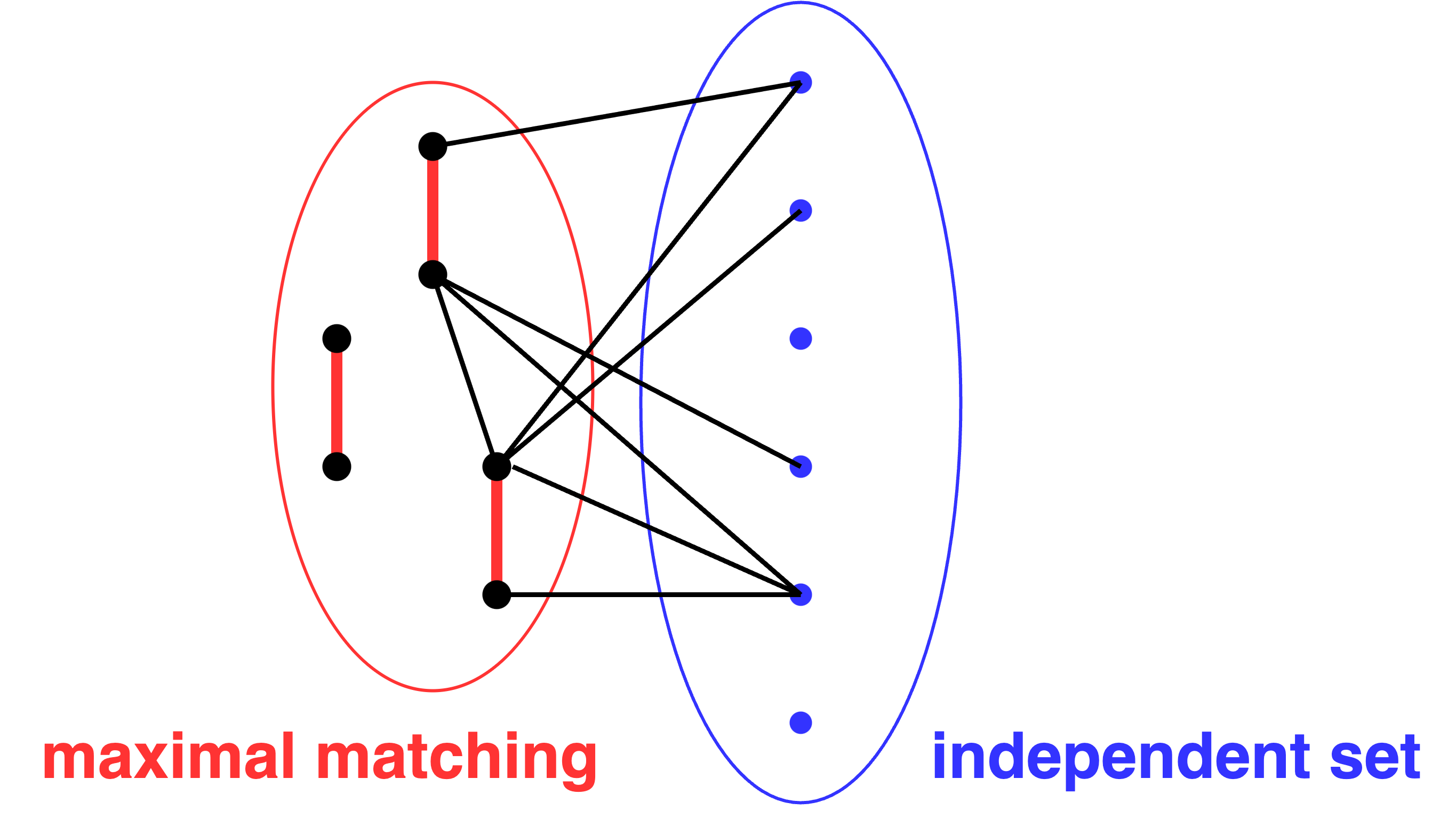}
\caption{The kernel for the $k$-vertex cover problem consists of maximal matching (with some missing vertices) and independent set.
The size of the maximal matching is at most $k$.
All non-missing vertices in the maximal matching have degree at most $k$.}
\label{fig:kernelization}
\end{figure}

\subsubsection{Overview of Our Parameterized Quantum Query Algorithm for Matching} \label{subsec:technical_overview_matching}

In our algorithm for the $k$-matching problem, we first find a matching of size at least $k$ or a maximal matching by Grover's search with $O(\sqrt{k} n)$ queries.
If we find a matching of size at least $k$, we conclude that the input instance $(G, k)$ is a yes-instance.
Now, we assume that we obtain a maximal matching $M$ of size at most $k - 1$.
Then, our algorithm relies on the theory of the augmenting path, a widely used approach for computing the maximum matching problem. 
We repeatedly find an $M$-augmenting path and update the current maximal matching $M$ until the size of $M$ is equal to $k$ or there is no $M$-augmenting path, which means that $M$ is a largest matching.

For a matching $M$, a path $P$ on $G$ is an $M$-augmenting path if the both endpoints of $P$ does not appear in $M$, and $P$ contains edges from $M$ and $E\setminus M$, alternatively.
It is easy to see that the symmetric difference $M \bigtriangleup E(P) = (M \setminus E(P)) \cup (E(P) \setminus M)$ is a matching of size $|M| + 1$.
It is well known that the size of a matching $M$ is largest if and only if there is no $M$-augmenting path (see e.g., \cite[Theorem 10.7]{korte2011combinatorial}).

In the each step of the algorithm, we have a maximal matching $M$ and all edges between vertices in $M$.
Then, we can classically find a candidate path that may be extended to an $M$-augmenting path by adding both endpoints.
Then, we check whether or not the candidate path can be extended to an $M$-augmenting path by applying Grover's search algorithm for finding neighborhoods of both the endpoints of the candidate path.
Hence, each step requires $O(\sqrt{n})$ queries.
For estimating the number of steps, we introduce the types of vertices and a potential function.

For each $v \in V(M)$, let $N_M(v)$ denote $\{u \in V(G) \setminus V(M) \mid \{v, u\} \in E(G) \}$.
In the algorithm, we classify each vertex $v$ in $V(M)$ into the following three types:
\begin{itemize}
 \item \texttt{type0} \quad We have determined that $N_M(v) = \emptyset$.
 \item \texttt{type1} \quad We have determined that $|N_M(v)| = 1$.
 \item \texttt{type2} \quad No information about $N_M(v)$.
\end{itemize}

At the beginning of the algorithm with some maximal matching $M$, all vertices in $V(M)$ are \texttt{type2}.
In the algorithms, the types of vertices in $V(M)$ are stored and updated.
For each vertex $v$ of \texttt{type1}, the unique neighborhood in $N_M(v)$ is stored as well.
Let $G[S]$ denote the subgraph of $G$ induced by $S\subseteq V(G)$.
For a maximal matching $M$ of graph $G$, a path $Q$ on $G[V(M)]$ is a candidate path if $Q$ satisfies the following conditions:
\begin{enumerate}
    \item The both ends of $Q$ are \texttt{type1} or \texttt{type2}. If the both are \texttt{type1}, their neighbors in $V(G)\setminus V(M)$ are different.
    \item $Q$ contains edges from $M$ and $E(G[V(M)])\setminus M$, alternatively.
    \item $Q$ contains edges from $M$ at the both ends.
\end{enumerate}
It is obvious that for any $M$-augmenting path $P$, the path $Q$ obtained by deleting the vertices at the both ends of $P$ is a candidate path.
In each step of the algorithm, we have a maximal matching $M$, the induced graph $G[V(M)]$ and the types of the vertices in $V(M)$, and search for a candidate path.
If there is no candidate path, we can conclude that the maximal matching $M$ is a maximum matching since there is no $M$-augmenting path.
In this case, the algorithm outputs NO and terminates.
Assume that there exists a candidate path $Q$.
If the both ends of $Q$ are \texttt{type1}, then $Q$ can be extended to an $M$-augmenting path.
If at least one of the ends of $Q$ is \texttt{type2}, we apply Grover's search algorithm for finding the neighbors in $V(G)\setminus V(M)$ of the vertices of \texttt{type2}.
If we find the different neighbors of the end vertices of $Q$, the candidate path $Q$ can be extended to a $M$-augmenting path.
If not, we can update the type of one of the \texttt{type2} vertex to \texttt{type0} or \texttt{type1}.
If an $M$-augmenting path $P$ is found, the maximal matching $M$ is updated to $M\bigtriangleup E(P)$ of size $|M|+1$.
In this case, the types of the two new vertices in the maximal matching are set to \texttt{type0} since $V(G)\setminus V(M)$ is an independent set.
We update the type of \texttt{type1} vertices that connect to either of the two new vertices to \texttt{type0}.
Finally, we queries classically for all the new pairs in $V(M)$ for obtaining the graph $G[V(M)]$.
This is the single step of the algorithm.
If the size of maximal matching $M$ achieves $k$, then the algorithm outputs YES, and terminates.

For each step, the size of maximal matching increases by one, or the number of vertices of \texttt{type2} decreases by one.
From this observation, we define the potential function $\displaystyle \Phi := k - |M| + N_2$ where $N_2$ is the number of vertices of \texttt{type2}.
Then, it is easy to see that $0 \leq \Phi \leq k - |M| + 2 |M| \leq 2 k$.
The above observation implies that the value of the potential function decreases by one at each step.
Hence, the number of iterations is at most $2k$. 

Finally, we estimate the query complexity of the quantum algorithm.
The query complexity for finding a maximal matching is $O(\sqrt{k}n)$.
The query complexity for finding neighbors of the both ends of a candidate path is $\tilde{O}(\sqrt{n})$ at each step.
The number of classical queries for obtaining all edges in $G[V(M)]$ is $O(k^2)$ in total.
Hence, the query complexity of the algorithm is $O(\sqrt{k}n) + \tilde{O}(k\sqrt{n}) + O(k^2) = O(\sqrt{k}n + k^2)$.

\subsection{Overview of Our Lower Bounds} \label{subsec:technical_overview_lower_bound}

Our lower bound is based on the adversary method introduced by Ambainis \cite{ambainis2000quantum}.
Although the technique is standard for proving lower bounds of quantum query complexity, we adapt this technique for the parameterized setting, and prove the lower bounds which matches to the upper bounds obtained in this paper for small $k$.

To obtain the lower bound for the {\sc $k$-vertex cover} problem, we consider the problem of distinguishing graphs that consist of $k$ disjoint edges and graphs that consist of $k + 1$ disjoint edges.
Any quantum algorithm for the {\sc $k$-vertex cover} problem must distinguish the above two classes of graphs.
By applying the adversary method for this problem, we obtain a lower bound $\Omega(\sqrt{k} n)$ when $k \leq n / 3$.

To prove a lower bound $\Omega(\sqrt{k} n)$ for any $k \leq (1 - \epsilon)n$, we consider generalized version of the problem.
In the problem, we distinguish graphs that consist of $\lfloor \frac{k}{t - 1} \rfloor$ cliques of size $t$ and graphs that consist of $\lceil \frac{k + 1}{t - 1} \rceil$ cliques of size $t$.
The proof of our lower bound for the {\sc $k$-matching} problem is similar to that for the $k$-vertex cover problem. 

\subsection{Additional Related Work} \label{subsec:related_works}

Ambainis-Balodis-Iraids-Kokainis-Pr{\=u}sis-Vihrovs \cite{ambainis2019quantum} and Le Gall-Seddighin \cite{le2023quantum} combine known standard quantum algorithms, such as Grover's algorithm, with classical techniques to design quantum algorithms with low time-complexities.

There are few previous works on the parameterized quantum query complexity. For the graph collision problem, we are given an undirected graph $G$ as a bit string and oracle access to Boolean variables $\{x_v \in \{0,1\}\mid v \in V(G) \}$ and, the objective is to decide whether there are two vertices $v$ and $u$ connected by an edge in $G$ such that $x_v = x_u = 1$.
Belovs \cite{belovs2012learning} showed that the quantum query complexity of this problem is $O(\sqrt{n}\alpha^{1/6})$, where $\alpha$ is the size of the largest independent set of $G$.
Ambainis--Balodis--Iraids--Ozols--Smotrovs--Juris \cite{ambainis2013parameterized} also showed that the quantum query complexity is $O(\sqrt{n} t^{1/6})$, where $t$ is the treewidth of $G$.
The only trivial $\Omega(\sqrt{n})$ lower bound is known for this problem.
We note that these studies did not use kernelization.

The usefulness of the concept of the parameterized complexity has been successfully demonstrated in other computational models, such as streaming model \cite{fafianie2014streaming, chitnis2016kernelization, chitnis2014parameterized, chitnis2019towards} and distributed model \cite{ben2019parameterized}, as well.
The study of parameterized streaming algorithms initiated by Fafianie--Kratsch~\cite{fafianie2014streaming} and Chitnis--Cormode--Hajiaghayi--Monemizadeh~\cite{chitnis2014parameterized} has been developed further by several authors.
In parameterized streaming algorithms, by analogy of FPT, Chitnis--Cormode~\cite{chitnis2019towards} introduced a hierarchy of space complexity classes and tight classification for several graph problems. 

There have also been many studies of the quantum query complexity of graph problems in the adjacency matrix model.
D{\"u}rr--Heiligman--H{\o}yer--Mhalla \cite{durr2006quantum} showed the quantum query complexity of testing connectivity is $\Theta(n^{3/2})$.
Ambainis--Iwama--Nakanishi--Nishimura--Raymond--Tani--Yamashita \cite{ambainis2008quantumISAAC} showed that the quantum query complexity of testing the planarity is also $\Theta(n^{3/2})$.
Belovs--Reichardt \cite{belovs2012span} presented a quantum algorithm that uses $O( \sqrt{d}n)$ queries to decide whether vertices $s$ and $t$ are connected, under the promise that they are connected by a path of length at most $d$, or are disconnected.
Apers--Lee \cite{apers2021quantum} presented a quantum algorithm to solve the weighted minimum cut problem using $O(n^{3/2}\sqrt{\tau})$ queries and time in the adjacency matrix model, where each edge weight is at least $1$ and at most $\tau$.

Now, we will mention some known results closely related to our work.
For bipartite graphs, the minimum vertex cover has the same cardinality as the maximum matching by K\"{o}nig's theorem (see e.g., \cite[Theorem 16.2]{schrijver2003combinatorial}).
Moreover, Blikstad--v.d.Brand--Efron--Mukhopadhyay--Nanongkai \cite{blikstad2022nearly} showed that the quantum query complexity of the bipartite maximum matching problem is $\tilde{O}(n^{3/2})$.
Then, for bipartite graphs, the quantum query complexity of the minimum vertex cover problem is $\tilde{O}(n^{3/2})$, which is almost optimal.

Childs--Kothari \cite{childs2012quantum} showed that any non-trivial minor-closed property that can not be described by a finite set of forbidden subgraphs has quantum query complexity $\Theta(n^{3/2})$.
On the other hand, they also show that any minor-closed properties that can be characterized by finitely many forbidden subgraphs can be solved in $O(n^{\alpha})$ queries for some $\alpha < 3/2$.
To show this, they proved that, for any constant $c>0$, a graph property that a graph either has more than $cn$ edges or contains a given subgraph $H$ can be decided in $\tilde{O}(n^{3/2-1/(\mathrm{vc}(H)+1)})$ queries, where $\mathrm{vc}(H)$ denote the size of minimum vertex cover of $H$.

Here, it is easy to verify that the property of having a vertex cover of size at most $k$ is  minor-closed. 
Dinneen--Lai \cite[Corollary 2.5]{dinneen2007properties} showed that the property of having a vertex cover of size at most $k$ can be characterized by finitely many forbidden subgraphs.
Then, the result of Childs--Kothari implies that the quantum query complexity of the {\sc $k$-vertex cover} problem is $\tilde{O}(n^{3/2-1/(k+2)})$ when $k$ is constant.


\section{Preliminaries} \label{sec:preliminaries}

\subsection{Basic Notation} \label{subsec:notation}

For a positive integer $p$, let $[p] = \{1, \dots, p\}$.
For a set $S$, let $\binom{S}{2}$ denote the set of all unordered pairs of elements in $S$.

In this paper, all graphs are undirected and simple.
For a graph $G$, let $V(G)$ and $E(G)$ denote the set of vertices and the set of edges in $G$, respectively.
For a matching $M \subseteq E(G)$ of $G$, let $V(M)$ denote the set of all endpoints of edges in $M$.
For a subset $S \subseteq V(G)$, let $G[S]$ denote the subgraph of $G$ induced by $S$.
We denote $G[V(G) \setminus S]$ briefly by $G - S$. 
We also write $G - v$ instead $G - \{ v \}$ for briefly.

\subsection{Query Complexity} \label{subsec:query_complexity}

In this paper, we consider the \emph{adjacency matrix model} for the input graph $G$.
In this model, the graph $G$ is given by an oracle $E_{\mathrm M} \colon \binom{V(G)}{2} \rightarrow \{0, 1\}$, where $E_{\mathrm M}(v, u) = 1$ if and only if $(v, u) \in E(G)$.

In the deterministic and randomized query complexity models, the input graph $G$ is accessed by querying the value of $E_{\mathrm M}(v, u)$.
In contrast, in the quantum model, an algorithm accesses the information of the input graph $G$ through a quantum oracle, which can be modeled by a unitary operator.
In the adjacency matrix model, a query oracle operator $O_{\mathrm M}$ acts as
$O_{\mathrm M} \ket{(v, u)} \ket{q} = \ket{(v, u)} \ket{q \oplus E_{\mathrm M}(v, u)}$, where $q \in \{0, 1\}$ and $\oplus$ denotes addition modulo $2$.

A quantum query algorithm to compute a function $F$ on a graph $G$ consists of a sequence of unitary operators $U_0,O_{\mathrm{M}},U_1,O_{\mathrm{M}},\dotsc,O_{\mathrm{M}},U_T$ where each $U_i$ is a general unitary operator that is independent of the input $G$.
Here, the oracle operators $O_\mathrm{M}$ may be applied to a part of the qubits.
The query complexity of the algorithm is equal to $T$, which is the number of quantum oracle calls in the algorithms.


For the randomized and quantum query complexity of computing a function $F$, we consider an algorithm that correctly outputs $F(G)$ with probability at least $2/3$ for every input $G$.

\subsection{Kernelization and Quantum Query Kernelization} \label{subsec:kernelization}

The technique of \emph{kernelization} is a general framework for designing of classical algorithms for parameterized problems.
A {\em parameterized problem} is defined as a language $L \subseteq \Sigma^* \times \mathbb{N}$, where $\Sigma$ is a fixed finite alphabet.
The kernelization is formally defined as follows.

\begin{definition}[Kernelization] \label{def:kernelization}
Let $L$ be a parameterized problem.
A {\em kernelization algorithm} is an algorithm that takes as input an instance $(x, k)\in\Sigma^*\times\mathbb{N}$ and outputs another instance $(x', k')\in\Sigma^*\times\mathbb{N}$ with running time polynomial in $|x|$ and $k$ such that 
\begin{itemize}
\item $(x, k) \in L$ if and only if $(x', k') \in L$, and
\item $|x'| + k'$ is bounded by $g(k)$ where $g \colon \mathbb{N} \rightarrow \mathbb{N}$ is some computable function.
\end{itemize}
The output $(x', k')$ is called a {\em kernel} of $(x, k)$.
\end{definition}



In this work, we introduce a new framework, called \emph{quantum query kernelization}, to obtain better quantum query complexities for parameterized graph problems.
Now, the quantum query kernelization is formally defined as follows.

\begin{definition}[Quantum Query Kernelization] \label{def:quantum_kernelization}
Let $L$ be a parameterized graph problem, in which we are given a parameter $k \in \mathbb{N}$ and a quantum oracle access to a graph $G$.
A {\em quantum query kernelization algorithm} is a quantum query algorithm that outputs another instance $(G', k')$ as a bit string such that $(G, k) \in L$ if and only if $(G', k') \in L$.
\end{definition}

By a quantum query kernelization algorithm, we obtain another instance $(G', k')$ as a bit string.
Then, we no longer require additional queries to the quantum oracle since we can decide whether $(G', k') \in L$ by an arbitrary classical algorithm.

\subsection{Tools for the Upper Bounds} \label{subsec:tools_for_algorithm}

In the proposed quantum algorithms, we use some variants of Grover's search algorithm \cite{grover1996fast}.
Let $f \colon [N] \rightarrow \{0, 1\}$ be a function.
We are given an oracle access to the function $f$.
Then, Grover's search algorithm finds an index $i \in [N]$ such that $f(i) = 1$ if one exists. 
Let $O_f$ be a quantum oracle defined by $O_f\ket{i}\ket{q} = \ket{i}\ket{q\oplus f(i)}$ for $i\in[N]$ and $q\in\{0,1\}$.



\begin{lemma}[Las Vegas-type Grover's Search algorithm without knowing the value of $|f^{-1}(1)|$ {\cite[Theorem 3]{boyer1998tight}}] \label{lem:grover_expected_find}
If $K:=|f^{-1}(1)|>0$,
Grover's search algorithm finds an index $i \in f^{-1}(1)$ with an expected number of quantum queries to $O_f$ at most $8 \sqrt{N / K}$.
Otherwise, it runs forever.
The output of the algorithm is chosen uniformly at random among the $K$ indices in $f^{-1}(1)$.
Note that this algorithm does not require prior knowledge of $K$.
\end{lemma}

To handle the case where there are no indices in $f^{-1}(1)$, we stop Grover's search algorithm at the appropriate time.
Then, Grover's search algorithm can determine that there are no indices in $f^{-1}(1)$ with bounded error by using $O(\sqrt{N})$ queries to $f$ \cite{boyer1998tight}.

We also use the following lemma.

\begin{lemma}[{\cite[Theorem 3]{buhrman1999bounds}}] \label{lem:grover_stop_approx}
Grover's search algorithm can determine that there are no indices in $f^{-1}(1)$ with probability at least $1 - \epsilon$ by using $O(\sqrt{N \log (1/\epsilon)})$ queries to $f$.
\end{lemma}

\section{Quantum Query Algorithm for Threshold Maximal Matching} \label{sec:algorithm_for_threshold_matching}

In this section, we provide a simple quantum query algorithm for the threshold maximal matching problem.
In this problem, we are given a graph $G$ and an integer $k$, and the objective is to find a matching of size at least $k + 1$ or a maximal matching.
Note that if there exists both a matching of size at least $k + 1$ and a maximal matching of size at most $k$, then we may output either of them.
As described in Sections~\ref{subsec:technical_overview} and \ref{subsec:technical_overview_matching}, this algorithm will play an important role in our quantum query algorithms for the {\sc $k$-vertex cover} and the {\sc $k$-matching} problems.

Our algorithm is similar to the quantum query algorithm for the maximal matching problem proposed by D{\"o}rn \cite{dorn2009quantum}.
D{\"o}rn show that the quantum query complexity of the maximal matching problem is $O(n^{3/2})$  \cite[Theorem 3]{dorn2009quantum}.
In our algorithm for the $k$-threshold maximal matching problem, we repeatedly apply the Las Vegas-type Grover's search algorithm, given in Lemma~\ref{lem:grover_expected_find}, to find an edge that can be added to the current matching $M$.
Our algorithm is described as \texttt{QuantumThresholdMaximalMatching}$[G](k)$ (Algorithm \ref{alg:quantum_threshold_maximal_matching}).
We note that our analysis of query complexity and error probability is similar to that of the $O(n^{3/2})$ upper bound for the connectivity problem by D{\"u}rr--Heiligman--H{\o}yer--Mhalla \cite[Theorem 5.1]{durr2006quantum}.





\begin{algorithm}[t]
    \KwInput{Oracle access to an unweighted graph $G$, and an integer $k$.}
    \KwOutput{A matching $M$}
    $M \gets \emptyset$ \\
    \While{the total number of queries used in Line \ref{line:quantum_maximal_matching_1} is less than $96 n \sqrt{k + 1}$} {
        Apply Las Vegas-type Grover's search algorithm to find an edge $e$ whose endpoints are both in $V(G) \setminus V(M)$  \label{line:quantum_maximal_matching_1} \\
        \tcp{When the total number of queries used in Line \ref{line:quantum_maximal_matching_1} exceeds $96 n \sqrt{k + 1} $, we stop the search algorithm in Line \ref{line:quantum_maximal_matching_1}.}
        $M \gets M \cup \{e\}$
    }

    
    \Return {$M$}
    \caption{\texttt{QuantumThresholdMaximalMatching}$[G](k)$}\label{alg:quantum_threshold_maximal_matching}
\end{algorithm}

\begin{algorithm}[t]
    $M \gets \emptyset$ \\
    \While{true} {
        Apply Las Vegas-type Grover's search algorithm to find an edge $e$ whose endpoints are both in $V(G) \setminus V(M)$ \\
        $M \gets M \cup \{e\}$
    }
    \caption{The procedure \texttt{QuantumThresholdMaximalMatching}$[G](k)$ (Algorithm \ref{alg:quantum_threshold_maximal_matching}) without the limit of queries} \label{alg:quantum_threshold_matching_run_forever}
\end{algorithm}


\begin{lemma} \label{lem:find_k_restricted_maximal_matching}
Given a graph $G$ with n vertices in the adjacency matrix model and an integer $k$,
the procedure \texttt{\textup{QuantumThresholdMaximalMatching}}$[G](k)$ (Algorithm \ref{alg:quantum_threshold_maximal_matching}) outputs a matching in $G$ and satisfies the following conditions:
\begin{itemize}
\item With probability at least $5/6$, the procedure outputs a matching of size at least $k + 1$ or a maximal matching.
\item The procedure requires $O(\sqrt{k} n)$ queries.
\end{itemize}
\end{lemma}

\begin{proof}[Proof of Lemma~\ref{lem:find_k_restricted_maximal_matching}]


We will prove that Algorithm~\ref{alg:quantum_threshold_maximal_matching} satisfies the conditions in Lemma.
We first consider Algorithm~\ref{alg:quantum_threshold_matching_run_forever}.
This algorithm does not terminate since once $M$ becomes a maximal matching the Las~Vegas-type Grover's search algorithm cannot find an edge $e$.
We will show that the expected number of queries until $M$ becomes a maximal matching or a matching of size $k+1$ is at most $16n\sqrt{k+1}$.
Then, Markov's inequality implies that even if we terminate the algorithm after $96n\sqrt{k+1}$ queries, the algorithm outputs a maximal matching or a matching of size at least $k+1$ with probability at least $5/6$.
This argument proves that Algorithm~\ref{alg:quantum_threshold_maximal_matching} satisfies the conditions in Lemma.

Recall that the Las~Vegas-type Grover's algorithm in Lemma~\ref{lem:grover_expected_find} has a finite expected number of queries. This means that the Las~Vegas-type Grover's algorithm finds an edge with probability 1 if it exists.
Hence, in the non-halting algorithm, $M$ becomes a maximal matching of $G$ after a finite number of queries with probability 1.
Let $M^*$ be a random variable corresponding to the maximal matching found by the non-halting algorithm.
Let $T=|M^*|$.
Then, by Lemma~\ref{lem:grover_expected_find}, for finding an $i$-th edge of $M^*$, the Las~Vegas-type Grover's algorithm has an expected quantum query complexity at most $8\sqrt{\frac{n^2}{T-i+1}}$ since there exist at least $T-i+1$ edges in vertex pairs of $V(G)\setminus V(M)$.
Hence, the expected number of queries until $M$ becomes a maximal matching or a matching of size $k+1$ is at most

\begin{align*}
\sum_{i = 1}^{\min\{T,\,k+1\}} 8 \sqrt{\frac{n^2}{T-i+1}} \leq
\sum_{i = 0}^{k} 8 \sqrt{\frac{n^2}{i+1}} \le
8 n \int_{0}^{k+1} \frac{dx}{\sqrt{x}} = 16 n \sqrt{k+1}.
\end{align*}
This completes the proof.
\end{proof}




\section{Parameterized Quantum Query Algorithm for Vertex Cover} \label{sec:algorithm}

In this section, we provide our quantum query algorithm for the {\sc $k$-vertex cover} problem.
Note that the essential idea and overview of the algorithm were given in Section~\ref{subsec:technical_overview}.
In this section, we design a quantum query algorithm for the finding problem, i.e., the algorithm that finds a vertex cover of size at most $k$ if exists.
To design an efficient quantum query algorithm for the {\sc $k$-vertex cover} problem, we apply our new approach, quantum query kernelization.
In this section, we consider a quantum algorithm that concludes that $G$ does not have a vertex cover of size at most $k$, or finds $U\subseteq V$ and a graph $G'$ such that the following properties hold:
\begin{enumerate}
\item If $G'$ has a vertex cover $S$ of size at most $k-|U|$, then $G$ has a vertex cover $U\cup S$ of size at most $k$,
\item If $G'$ does not have a vertex cover of size at most $k-|U|$, then $G$ does not have a vertex cover of size at most $k$,
\item The number of edges in $G'$ is $O(k^2)$.
\end{enumerate}
We emphasize that $G'$ must be obtained as a bit string.
If we only want to solve the {\sc $k$-vertex cover problem} as the decision problem,
we do not have to compute $U$ and sufficient to output $G'$ and $k'=k-|U|$ so that $(G',k')$ is a kernel of $(G,k)$.

In our quantum query kernelization algorithm, we first apply the procedure \\ \texttt{QuantumThresholdMaximalMatching}$[G](k)$ given in Section~\ref{sec:algorithm_for_threshold_matching}.
This procedure outputs a maximal matching $M$ or a matching of size at least $k + 1$ with $O(\sqrt{k}n)$ queries.
If $G$ has a matching of size at least $k + 1$, then we conclude that $G$ does not have a vertex cover of size at most $k$.

Otherwise, we compute $U$ and $G'$ as a bit string in the following way.
By Grover's search algorithm, we find all edges incident to vertices in $V(M)$.
In this procedure, if some vertex $v\in V(M)$ has at least $k+1$ neighbors, then $v$ is added to $U$.
Then, we do not search for edges incident to $U$.
Finally, we obtain $U\subseteq V(M)$ being a set of vertices of degree at least $k+1$ and $G'$ consisting of all edges between $V(M)\setminus U$ and $V(G)\setminus U$ and all vertices with degree at least 1.
Since all vertices in $V(M)\setminus U$ have degree at most $k$, the number of edges in $G'$ is at most $2k^2$.
%
Then, we obtain a desired output $U$ and $G'$ as a bit string.
After we obtain the kernel, we do not require more queries to the quantum oracle; we just have to apply a classical algorithm for finding a vertex cover $S$ of size at most $k-|U|$ in $G'$.
Our algorithm is described as \texttt{QuantumVertexCover}$[G](k)$ in Algorithm~\ref{alg:quantum_vertex_cover}.

\begin{algorithm}[t]
    \KwInput{Oracle access to an unweighted graph $G$, and an integer $k$.}
    \KwOutput{Find a vertex cover of size at most $k$, or conclude that the input graph $G$ does not have a vertex cover of size at most $k$.}
    $M \gets$ \texttt{QuantumThresholdMaximalMatching}$[G](k)$ \label{line:quantum_kernelization_1} \\
    \If{$|M| > k$} {
        \Return Conclude $G$ does not contain a vertex cover of size at most $k$.
    }
    $U \gets \emptyset$, $E' \gets \emptyset$\\
    \For{$v \in V(M)$}{ $d_v \gets 0$ }
    \While{the total number of queries used in Line \ref{line:quantum_vertex_cover_2} is less than $192 k \sqrt{(k + 1) n}$} {
        Apply Las Vegas-type Grover's search algorithm to find an edge $(v,u) \in \{ (v, u) \in E(G) \mid v \in V(M) \setminus U \text{ and } u \in V(G) \setminus U \} \setminus E'$  \label{line:quantum_vertex_cover_2} \\
        \tcp{When the total number of queries used in Line \ref{line:quantum_vertex_cover_2} exceeds $192 k \sqrt{(k + 1) n}$, we stop the search algorithm in Line \ref{line:quantum_vertex_cover_2}.}
        $E' \gets E' \cup \{ (v, u) \}$ \\
        $d_v \gets d_v + 1$ \\
        \If{$d_v > k$} {
            $U \gets U \cup \{v\}$\\
        }
        \If{$u \in V(M)$} {
            $d_u \gets d_u+1$ \\
            \If{$d_u > k$} {
                $U \gets U \cup \{u\}$\\
            }
        }
    }
    $G' \gets (V(G), E')$ \\
    $G' \gets G'[V(G')\setminus U]$\\
    Solve the $k$-vertex cover problem for the instance $(G',  k-|U|)$ by an arbitrary classical algorithm. \label{line:quantum_vertex_cover_3} \\
    \If{the classical algorithm finds a vertex cover $S$ in $G'$} {
        \Return $S \cup U$.
    } \Else {
        \Return Conclude that $G$ does not have a vertex cover of size at most $k$.
    }
    \caption{\texttt{QuantumVertexCover}$[G](k)$}\label{alg:quantum_vertex_cover}
\end{algorithm}

We show the following theorem, which implies Theorem \ref{thm:main_theorem_matrix_model}.

\begin{theorem} \label{thm:quantum_query_best_algorithm_vertex_cover}
Given a graph $G$ with n vertices in the adjacency matrix model and an integer $k$, with probability at least $2/3$, the procedure \textup{\texttt{QuantumVertexCover}}$[G](k)$ (Algorithm~\ref{alg:quantum_vertex_cover}) finds a vertex cover of size at most $k$ or otherwise determines that there does not exist such a vertex cover.
The query complexity of the algorithm is $O(\sqrt{k}n + k^{3/2}\sqrt{n})$.
\end{theorem}
By combining this with the lower bound given in Theorem \ref{thm:main_theorem_matrix_model_lower_bound}, we observe that this upper bound is tight when $k = O(\sqrt{n})$.

For the upper bound of the query complexity in Theorem~\ref{thm:quantum_query_best_algorithm_vertex_cover}, we first show the following lemma.

\begin{lemma} \label{lem:vertex_cover_while_loop}
Suppose that \texttt{\textup{QuantumThresholdMaximalMatching}}$[G](k)$ in Line \ref{line:quantum_kernelization_1} finds a maximal matching $M$ of size at most $k$. 
Then, after the execution of the while loop,
the following properties hold with probability at least $5/6$.
\begin{itemize}
\item[(i)] For all $v \in V(M)$, if the degree of a vertex $v \in V(M)$ in $G$ is at least $k + 1$, then $v \in U$.
\item[(ii)] For all $v \in V(M)$, if the degree of a vertex $v \in V(M)$ in $G$ is at most $k$, then $(v, u) \in E'$ for all $(v, u) \in E(G)$. 
\end{itemize}
\end{lemma}

\begin{proof}
In the same way as the proof of Lemma \ref{lem:find_k_restricted_maximal_matching}, we consider the while loop of the procedure \texttt{QuantumVertexCover} without the limit of queries, which is described as Algorithm \ref{alg:quantum_vertex_cover_run_forever}.
We will show that the expected total number of queries until both the properties (i) and (ii) hold is at most $32 k \sqrt{(k + 1) n}$.
Then, in Algorithm \ref{alg:quantum_vertex_cover_run_forever}, after the Grover's search algorithm used $192 k \sqrt{(k + 1) n}$ queries, both the properties (i) and (ii) hold with probability at least $5/6$.


\begin{algorithm}[t]
    $U \gets \emptyset$, $E' \gets \emptyset$ \\
    \For{$v \in V(M)$}{ $d_v \gets 0$ }
    \While{true} {
        Apply Las Vegas-type Grover's search algorithm to find an edge $(v,u) \in \{ (v, u) \in E(G) \mid v \in V(M) \setminus U \text{ and } u \in V(G) \setminus U \} \setminus E'$\\
        $E' \gets E' \cup \{ (v, u) \}$ \\
        $d_v \gets d_v + 1$ \\
        \If{$d_v > k$} {
            $U \gets U \cup \{v\}$\\
        }
        \If{$u \in V(M)$} {
            $d_u \gets d_u+1$ \\
            \If{$d_u > k$} {
                $U \gets U \cup \{u\}$\\
            }
        }
    }
    \caption{The while loop of the procedure \texttt{QuantumVertexCover}$[G](k)$ (Algorithm \ref{alg:quantum_vertex_cover}) without the limit of queries} \label{alg:quantum_vertex_cover_run_forever}
\end{algorithm}

Let $E'^{*}$ (resp. $U^*$) be a random variable corresponding to $E'$ (resp. $U$) found by the non-halting algorithm.
Let $T = |E'^{*}|$.
Then, for finding an $i$-th edge of $E'^{*}$, the Las Vegas-type Grover's algorithm has an expected quantum query complexity at most $8 \sqrt{\frac{2 k n}{T - i + 1}}$, since there exist at least $T - i + 1$ edges in the set $\{ (v, u) \in E(G) \mid v \in V(M) \setminus U \text{ and } u \in V(G) \} \setminus E'$.
Note that we have $|V(M) \setminus U| \leq 2 k$, and then the size of the vertex pair set $\{ (v, u) \mid v \in V(M) \setminus U \text{ and } u \in V(G) \}$ is at most $2 k n$.
Hence, the expected number of queries until both properties (i) and (ii) hold is at most
\begin{align*}
\sum_{i = 1}^{T} 8 \sqrt{\frac{2 k n}{T - i + 1}} \leq 8 \sqrt{2 k n} \int_{0}^{T} \frac{dx}{\sqrt{x}} = 16 \sqrt{2 k n T}.
\end{align*}
Here, $T \leq |V(M)| \cdot (k + 1) \leq 2 k (k + 1)$.
Thus, the expected number of queries until both the properties (i) and (ii) hold is at most $32 k \sqrt{(k + 1) n}$.
This completes the proof.
\end{proof}

Now we provide a proof of Theorem \ref{thm:quantum_query_best_algorithm_vertex_cover}.

\begin{proof}[Proof of Theorem \ref{thm:quantum_query_best_algorithm_vertex_cover}]
From Lemma~\ref{lem:vertex_cover_while_loop}, the while loop computes $U$ and $E'$ correctly with probability at least $5/6$.
The query complexities of \texttt{QuantumThresholdMaximalMatching}$[G](k)$ and the while loop are $O(\sqrt{k}n)$ and $O(k^{3/2}\sqrt{n})$, respectively.
Hence, the total number of queries is $O(\sqrt{k}n + k^{3/2}\sqrt{n})$.
\end{proof}

\section{Parameterized Quantum Query Algorithm for Matching} \label{sec:algorithm_for_matching}

In this section, we provide our quantum query algorithms for the {\sc $k$-matching} problem and the maximum matching problem.
The overview of the algorithm was given in Section~\ref{subsec:technical_overview_matching}.

\subsection{Quantum Query Algorithm for the \texorpdfstring{$k$}{TEXT}-Matching Problem} \label{subsec:quantum query_algo_for_k_matching}

In our algorithm, we first apply the procedure \texttt{QuantumThresholdMaximalMatching}$[G](k - 1)$ given in Section \ref{sec:algorithm_for_threshold_matching}.
This procedure outputs a maximal matching or a matching of size at least $k$ with $O(\sqrt{k}n)$ queries.
If this procedure outputs a matching of size at least $k$, then we conclude that $G$ has a matching of size at least $k$.
If the procedure outputs a maximal matching $M$, we repeatedly search for an $M$-augmenting path and update the current matching $M$ until its size is at least $k$ or there is no $M$-augmenting path, i.e., $M$ is the largest matching.
Note that the current matching $M$ is always maximal.

In the algorithm, we store all edges $F$ in $G[V(M)]$ as a bit string.
We initialize $F = \emptyset$.
For all $(v, u) \in \binom{V(M)}{2}$, we simply check whether $(v, u) \in E(G)$ by classical queries, and if it does, we add $(u, v)$ to $F$. 
Now, we have $G[V(M)] = (V(M), F)$.
After updating the current matching $M$ by an $M$-augmenting path $s, v_1, \ldots, v_l, t$ as $M\leftarrow M\bigtriangleup \{(s,v_1),(v_1,v_2),\dotsc,(v_l,t)\}$, we check whether $(s, u) \in E(G)$ and $(t, u) \in E(G)$ for all $u \in V(M) \setminus \{s, t \}$ by classical queries, and if it does, we add the edges to $F$. 
We repeat this procedure until $M$ has size $k$ or there is no $M$-augmenting path.
This is the abstract framework of the algorithm.

The validity of the algorithm is obvious.
We will analyze the query complexity of this algorithm.
The query complexity for finding a maximal matching $M$ is $O(\sqrt{k}n)$.
The query complexity for finding all edges in $G[V(M)]$ is $O(k^2)$
since $|V(M)|\le 2k$.

For reducing the query complexity for the augmenting part, we maintain types of vertices in $V(M)$.
For each $v \in V(M)$, let $N_M(v)$ denote $\{u \in V(G) \setminus V(M) \mid (v, u) \in E(G) \}$.
In our algorithm, we classify each vertex $v$ in $V(M)$ into the following three {\em types}:


\begin{itemize}
 \item \texttt{type0} \quad We have determined that $N_M(v) = \emptyset$.
 \item \texttt{type1} \quad We have determined that $|N_M(v)| = 1$.
  \item \texttt{type2} \quad No information about $N_M(v)$.
\end{itemize}

At the beginning of the augmenting part, all vertices in $V(M)$ are \texttt{type2}.
In the repetitions, the types of vertices are updated from \texttt{type2} to \texttt{type0} or \texttt{type1}.
For vertices of \texttt{type1}, the unique neighborhood in $V(G)\setminus V(M)$ is stored.
For $v \in V(M)$, let $\texttt{type}(v)$ denote the type of $v$, and for $v\in V(M)$ of \texttt{type1}, $\texttt{memo}(v)$ denote the unique neighborhood of $v$ in $V(G)\setminus V(M)$.

Recall that a path $P = s, v_1, \ldots, v_l, t$ is an $M$-augmenting path if $P$ has odd length, $M$ does not touch both endpoints of $P$, and its edges are alternatingly out of and in $M$.
Here, we have $v_i \in V(M)$ for all $i \in [l]$ and $s, t \in V(G) \setminus V(M)$.
Since we store all edges $F$ in $G[V(M)]$ as a bit string, we can find a path $v_1, \ldots, v_l$ that is part of an $M$-augmenting path $P$.
Such a path is called a {\em candidate path}, which is formally defined as follows.

\begin{definition} \label{def:candidate_path}
A path $v_1, \ldots, v_l$ is said to be a \emph{candidate path} 
if 
\begin{itemize}
 \item $l$ is even,
 \item $v_i \in V(M)$ for all $i \in [l]$,
 \item $(v_i, v_{i + 1}) \in M$ for all $i \in [l - 1]$ such that $i$ is odd,
 \item $(v_i, v_{i + 1}) \in E(G) \setminus M$ for all $i \in [l - 1]$ such that $i$ is even,
 \item neither $v_1$ nor $v_l$ is \textup{\texttt{type0}}, and
 \item if both $v_1$ and $v_l$ are \textup{\texttt{type1}}, then $N_M(v_1) \neq N_M(v_l)$.
\end{itemize}
\end{definition}

\begin{figure}[t]
\centering
\includegraphics[width=6cm]{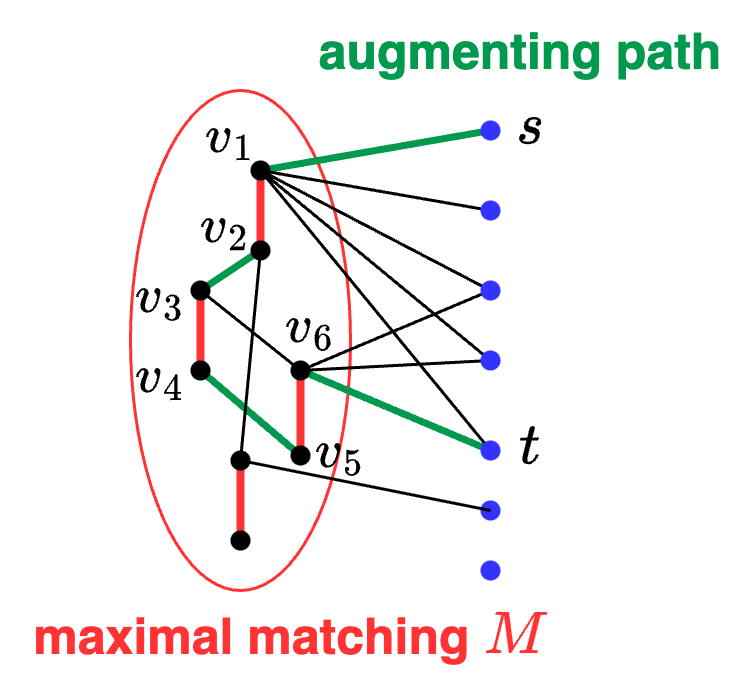}
\caption{Example of an $M$-augmenting path and a candidate path}
\label{fig:candidate_path}
\end{figure}

We give an example of an $M$-augmenting path and a candidate path in Figure \ref{fig:candidate_path}.
In this figure, there are an $M$-augmenting path $s, v_1, \ldots, v_6, t$ and a candidate path $v_1, \ldots, v_6$.
It is obvious that for any $M$-augmenting path $P$, the path $Q$ obtained by deleting the vertices at the both ends of $P$ is a candidate path.

In each step of the augmenting part, we have a maximal matching $M$, the induced subgraph $G[V(M)]$ and the type of vertices in $V(M)$.
Then, we search for candidate paths.
If there is no candidate path, we can conclude that the maximal matching $M$ is a maximum matching since there is no $M$-augmenting path.
If there exists a candidate path $Q$, we check whether $Q$ can be extended to an $M$-augmenting path.


\begin{lemma} \label{lem:one_iteration_algo}
Given a graph $G$ with n vertices in the adjacency matrix model, $\epsilon>0$, a matching $M$, a candidate path $Q$, \textup{\texttt{type}}, and \textup{\texttt{memo}},
the procedure \textup{\texttt{ExtendCandidatePath}}$[G](\epsilon, M, Q, \textup{\texttt{type}}, \textup{\texttt{memo}})$ in Algorithm~\ref{alg:update_matching_or_types} satisfies the following conditions:
\begin{itemize}
\item The algorithm updates the type of a vertex in $V(M)$ of \textup{\texttt{type2}} to \textup{\texttt{type0}} or \textup{\texttt{type1}}, or finds an $M$-augmenting path.
\item The query complexity is $O(\sqrt{n \log (1/\epsilon)})$.
\item The error probability is $\epsilon$.
\end{itemize}
\end{lemma}
\begin{proof}

We denote by $Q = v_1, v_2, \dots , v_{l - 1}, v_l$.
Then, we can check whether $Q$ can be extended to an $M$-augmenting path by an appropriate application of Grover's search according to the types of $v_1$ and $v_l$.

\begin{itemize}
\item Case 1: Both $v_1$ and $v_l$ are \textup{\texttt{type1}}.\\
Let $N_M(v_1) = \{s \}$ and $N_M(v_l) = \{ t \}$.
By the definition of a candidate path, we have $s \neq t$.
Thus, a path $s, v_1, \ldots , v_l, t$ is an $M$-augmenting path.

\item Case 2: Both $v_1$ and $v_l$ are \textup{\texttt{type2}}.\\
We apply Grover's search algorithm, given in Lemma~\ref{lem:grover_stop_approx}, that requires $O(\sqrt{n \log (1/\epsilon)})$ queries to find a vertex $s$ in $V(G) \setminus V(M)$ that is adjacent to $v_1$.
If we  do not find a such vertex $s$, then we update the type of $v_1$ to \texttt{type0}.
Otherwise, we apply Grover's search algorithm that requires $O(\sqrt{n \log (1/\epsilon)})$ queries to find a vertex $t$ in $V(G) \setminus (V(M) \cup \{ s \})$ that is adjacent to $v_l$.
If we find such a vertex $t$, a path $s, v_1, \ldots , v_l, t$ is an $M$-augmenting path.
Otherwise, we check whether $(v_l, s) \in E(G)$.
Then, we update the type of $v_l$ to \texttt{type1} and let $\mathtt{memo}(v_l)=s$ if $(v_l, s) \in E(G)$, and update the type of $v_l$ to \texttt{type0} otherwise.

\item Case 3: The other case. \\
By the definition of a candidate path, one of $v_1$ and $v_l$ is \texttt{type1}, and the other is \texttt{type2}.
Without loss of generality, we assume $v_1$ is \texttt{type1} and $v_l$ is \texttt{type2}.
Let $N_M(v_1) = \{s \}$.
We apply Grover's search algorithm that requires $O(\sqrt{n \log (1/\epsilon)})$ queries to find a vertex $t$ in $V(G) \setminus (V(M) \cup \{ s \})$ that is adjacent to $v_l$.
If we find such a vertex $t$, a path $s, v_1, \ldots , v_l, t$ is an $M$-augmenting path.
Otherwise, we check whether $(v_l, s) \in E(G)$.
Then, we update the type of $v_l$ to \texttt{type1} and let $\mathtt{memo}(v_l)=s$ if $(v_l, s) \in E(G)$, and update the type of $v_l$ to \texttt{type0} otherwise.
\end{itemize}

By Lemma \ref{lem:grover_stop_approx}, the error probability of the procedure \textup{\texttt{ExtendCandidatePath}} is $O(\epsilon)$, which completes the proof.
\end{proof}
\begin{algorithm}[H]
    \KwInput{Oracle access to a graph $G$, $\epsilon>0$, a matching $M$, a candidate path $Q$, \texttt{type}, and \texttt{memo}.}
    \KwOutput{An $M$-augmenting path $P$ if $Q$ can be extended to $P$ or ``TYPE UPDATED" otherwise.}
    \tcp{Denote by $Q = v_1, v_2, \dots , v_{l - 1}, v_l$}
    \If{both $v_1$ and $v_l$ are \textup{\texttt{type1}}} {
        $s \gets \texttt{memo}(v_1)$, $t \gets \texttt{memo}(v_l)$ \\
        \Return the $M$-augmenting path $P = s, v_1, \ldots, v_l, t$
    } \ElseIf {both $v_1$ and $v_l$ are \textup{\texttt{type2}}} {
        Apply Grover's search algorithm that uses $O(\sqrt{n \log (1/\epsilon)})$ queries to find a vertex $s$ in $V(G) \setminus V(M)$ that is adjacent to $v_1$ \\
        \If{there does not exist such a vertex $s$} {
            \texttt{type}$(v_1) \gets 0$ \\
            \Return "TYPE UPDATED"
        }
        Apply Grover's search algorithm that uses $O(\sqrt{n \log (1/\epsilon)})$ queries to find a vertex $t$ in $V(G) \setminus (V(M) \cup \{ s \})$ that is adjacent to $v_l$ \\
        \If{there does not exist such a vertex $t$} {
            \lIf{$(v_l, s) \in E(G)$ is confirmed by a query}{
                \texttt{type}$(v_l) \gets 1$, \texttt{memo}$(v_l) \gets s$
            }
            \lElse{
                \texttt{type}$(v_l) \gets 0$
            }
            \Return "TYPE UPDATED"
        }
        \Return the $M$-augmenting path $P = s, v_1, \ldots, v_l, t$
    } \Else {
        Without loss of generality, we assume $v_1$ is \texttt{type1} and $v_l$ is \texttt{type2} \\
        $s \gets \texttt{memo}(v_1)$ \\
        Apply Grover's search algorithm that uses $O(\sqrt{n \log (1/\epsilon)})$ queries to find a vertex $t$ in $V(G) \setminus (V(M) \cup \{ s \})$ that is adjacent to $v_l$ \\
        \If{there does not exist such a vertex $t$} {
            \lIf{$(v_l, s) \in E(G)$ is confirmed by a query}{
                \texttt{type}$(v_l) \gets 1$, \texttt{memo}$(v_l) \gets s$
            }
            \lElse{
                \texttt{type}$(v_l) \gets 0$
            }
            \Return "TYPE UPDATED"
        }
        \Return the $M$-augmenting path $P = s, v_1, \ldots, v_l, t$
    }
    \caption{\texttt{ExtendCandidatePath}$[G](\epsilon, M, Q, \texttt{type}, \texttt{memo})$}\label{alg:update_matching_or_types}
\end{algorithm}

The quantum algorithm for the {\sc $k$-matching} problem is described as \texttt{Quantum $k$-Matching}$[G](k)$ in Algorithm~\ref{alg:quantum_matching}.
We show the following theorem, which implies Theorem \ref{thm:main_theorem_parameterized_matching_matrix_model}.

\begin{theorem} \label{thm:upper_bound_of_k_matching}
Given a graph $G$ with n vertices in the adjacency matrix model and an integer $k$, with probability at least $2/3$, the procedure \textup{\texttt{Quantum $k$-Matching}}$[G](k)$ in Algorithm~\ref{alg:quantum_matching} finds a matching of size at least $k$ or 
 otherwise determines that there does not exist such a matching.
The query complexity of this procedure is $O(\sqrt{k}n + k^{2})$.
\end{theorem}
\begin{proof}
In the procedure \texttt{Quantum $k$-Matching}, we first apply the procedure \\ \texttt{QuantumThresholdMaximalMatching}$[G](k - 1)$, which, by Lemma \ref{lem:find_k_restricted_maximal_matching}, uses $O(\sqrt{k} n)$ queries.

We  need to store all edges $F$ in $G[V(M)]$ as a bit string.
In the whole run of \texttt{Quantum $k$-Matching}, we require $O(k^2)$ queries to obtain all edges in $G[V(M)]$ in Lines 7 and 22 of Algorithm~\ref{alg:quantum_matching}.
This is because the number of all pairs $(v, u) \in \binom{V(M)}{2}$ is $O(k^2)$ since $|V(M)| \leq 2k$.

Next, we repeatedly apply the procedure \texttt{ExtendCandidatePath} until there is no $M$-augmenting path or we have $|M| \geq k$.
By Lemmas \ref{lem:one_iteration_algo}, the query complexity of a single call of \texttt{ExtendCandidatePath} is $O(\sqrt{n \log k})$.
We will show that the number of repetitions is $O(k)$.
We define the potential function
\begin{align*}
\Phi &:= k-|M| + N_2
\end{align*}
where $N_2$ is the number of vertices in $V(M)$ of \texttt{type2}.
Then, $0\le\Phi\le k-|M|+2|M|\le k+|M|\le 2k$.
By a single call of \texttt{ExtendCandidatePath}, the value of the potential function decreases by one.
Hence, the number of calls of \texttt{ExtendCandidatePath} is at most $2k$.
Hence, the total query complexities due to \texttt{ExtendCandidatePath} is $O(k\sqrt{n\log k})$

Note that $k \sqrt{n \log (k)} = o(\sqrt{k}n + k^2)$.
Therefore, the query complexity of \texttt{Quantum $k$-Matching} is $O(\sqrt{k}n + k^2)$.

Now, it remains to bound the error probability of \texttt{Quantum $k$-Matching}.
By Lemma \ref{lem:find_k_restricted_maximal_matching}, the error probability of \texttt{QuantumThresholdMaximalMatching} is at most $1/6$.
By Lemma \ref{lem:one_iteration_algo}, the error probability of one execution of \texttt{ExtendCandidatePath} is $1/(12k)$.
The number of executions of \texttt{ExtendCandidatePath} is at most $2k$.
Then, by taking the union bound over the whole run of \texttt{Quantum $k$-Matching}, the error probability by \texttt{ExtendCandidatePath} is at most $1/6$.
Therefore, the procedure \texttt{Quantum $k$-Matching} fails with probability at most $1/3$, which completes the proof.
\end{proof}

\begin{algorithm}[H]
    \KwInput{Oracle access to a graph $G$, and an integer $k$.}
    \KwOutput{Find a matching of size at least $k$, or conclude that $G$ does not have a matching of size $k$.}
    $M \gets$ \texttt{QuantumThresholdMaximalMatching}$[G](k - 1)$ \label{line:quantum_matching_0} \\
    \If{$|M| \geq k$} {
        \Return {$M$}
    }
    $F \gets \emptyset$ \\
    \For{$v \in V(M)$} {
        \For{$u \in V(M)$} {
            \If{$(v, u) \in E(G)$ is confirmed by a query} { \label{line:quantum_maching_1}
                $F \gets F \cup \{(v, u)\}$ \\
            }
        }
    }
    For all $v \in V(M)$ let $\texttt{type}(v) \gets 2$ \\
    \While {$|M| < k$} { \label{line:qunatum_matching_while_loop}
        $Q\gets$ a candidate path computed by a classical algorithm\\
        \If{$Q$ does not exist} {
            \Return ``$G$ does not have a matching of size $k$''.
        }
        $P \gets$ \texttt{ExtendCandidatePath}$[G](1/(12k), M, Q, \texttt{type}, \texttt{memo})$ (Algorithm \ref{alg:update_matching_or_types}) \\
        \If{$P =$ \textup{"TYPE UPDATED"}} {
            \Continue
        }
        $M \gets M \bigtriangleup E(P)$ \label{line:quantum_matching_6} \\
        $s,\, t\gets \text{The vertices at the ends of }P$\\
        $\texttt{type}(s) \gets 0$, $\texttt{type}(t) \gets 0$ \\
        \For{$v \in \{s, t \}$} {
            \For{$u \in V(M) \setminus \{s, t\}$} { \label{line:quantum_matching_update_start}
                \If{$(v, u) \in E(G)$ is confirmed by a query} { \label{line:quantum_maching_3}
                    $F \gets F \cup \{(v, u)\}$ \\
                }
            }
        }
        \For{$v \in V(M)$ with \textup{\texttt{type}}$(v) = 1$} {
            \If{\textup{\texttt{memo}}$(v) = s$ or $t$ } {
                \texttt{type}$(v) \gets 0$ \\
            }
        }
    }
    \Return {$M$}
    \caption{\texttt{Quantum $k$-Matching}$[G](k)$}\label{alg:quantum_matching}
\end{algorithm}

\subsection{Quantum Query Algorithm for the Maximum Matching Problem} \label{subsec:quantum query_algo_for_max_matching}

In this subsection, we provide our quantum query algorithm for the maximum matching problem.
Let $p$ denote the size of the maximum matching in the input graph $G$.
To find a maximum matching of size $p$, we set $k = 2, 4, \ldots, 2^{\lceil \log p \rceil}$ and apply the algorithm for the {\sc  $k$-matching} problem.
Recall that the procedure \texttt{Quantum $k$-Matching} (Algorithm \ref{alg:quantum_matching}) finds a matching of size at least $k$ or concludes that there is no such a matching.
Furthermore, when this procedure concludes that there is no such a matching, it finds a maximum matching of size at most $k - 1$.
In this method, we need to use $O((\sqrt{p} n + p^2) \cdot \log(\log(p)))$ queries.
This upper bound contains a factor polynomial in $\log n$ to bound the error probability.
To obtain an upper bound without a factor polynomial in $\log n$, we use the algorithm given in Lemma \ref{lem:quantum_maximal_matching_via_decision_tree_method} (see Section \ref{appendix:desicion_tree_method} for this lemma).
This algorithm can find a maximum matching with $O(\sqrt{p}n)$ queries without knowing the value of $p$ ahead.

\begin{proof}[Proof of Theorem \ref{thm:main_theorem_matching_matrix_model}]
By Lemma \ref{lem:quantum_maximal_matching_via_decision_tree_method}, we can find a maximal matching with $O(\sqrt{p}n)$ queries.
Here, let $\bar{p}$ denote the size of the obtained maximal matching.
We have $p \leq 2 \bar{p}$, since any maximal matching is at least half the size $1/2$ of a maximum matching.
Then, we apply the procedure \texttt{Quantum $k$-Matching} for $k = 2\bar{p}$, and we can find a maximum matching.
By Theorem \ref{thm:upper_bound_of_k_matching}, the total number of queries is $O(\sqrt{p}n + p^2)$, which completes the proof.
\end{proof}

\subsection{Quantum Query Algorithm for Maximal Matching via Decision Tree Method} \label{appendix:desicion_tree_method}


To prove Theorem \ref{thm:main_theorem_matching_matrix_model}, we need to show the following lemma.
Let $p$ be the size of the maximum matching in the input graph $G$.
Note that D{\"o}rn~\cite{dorn2009quantum} and Beigi--Taghavi--Tajdini~\cite{beigi2022time} show that the quantum query complexity of the maximal matching problem is $O(n^{3/2})$.

\begin{lemma} \label{lem:quantum_maximal_matching_via_decision_tree_method}
There is a bounded error quantum algorithm that finds a maximal matching with $O(\sqrt{p} n)$ queries without knowing the value of $p$ ahead.
\end{lemma}

To show this lemma, we use the decision tree method developed by Lin--Lin \cite{lin2015upper} and Beigi--Taghavi \cite{beigi2020quantum}.
Lin--Lin developed a remarkable tool to obtain an upper bound on the quantum query complexity of computing a function $F \colon \{0, 1\}^n \rightarrow [m]$ from a classical query algorithm of computing $F$.
Here, the input of the function $F$ can be accessed via queries to its bits.
Beigi--Taghavi generalized Lin--Lin's result for a function $F \colon [l]^n \rightarrow [m]$ with non-binary input.


A classical query algorithm can be represented as a decision tree.
In a decision tree,
each internal vertex represents a query, and each outgoing edge from a vertex represents a possible outcome of the query.
Multiple query outcomes can be combined into a single edge if the algorithm's future decisions are not dependent on which query outcome occurred.
Each leaf vertex is labeled by an element in $[m]$, which corresponds to an output of the algorithm.
Given such a decision tree, we can construct a guessing scheme.
In a guessing scheme, we choose exactly one outgoing edge from each vertex as the guess.
If the guess matches the outcome of the query, we call it a correct guess. 
Otherwise, we call it an incorrect guess.

\begin{theorem}[Decision tree method {\cite[Theorem 4]{beigi2020quantum}}] \label{thm:guessing_tree_algorithm}
Given a decision tree $\T$, for a function $F \colon D_F \rightarrow [m]$ with $D_F \subseteq [l]^n$, and a guessing scheme of $\T$.
Let $T$ be the depth of $\T$.
Let $I$ be the maximum number of incorrect guesses in any path from the root to a leaf of $\T$.
Then, the quantum query complexity of computing the function $F$  with bounded error is $O(\sqrt{T I})$.
\end{theorem}

\begin{proof}[Proof of Lemma \ref{lem:quantum_maximal_matching_via_decision_tree_method}]
To create a decision tree, we use a simple greedy algorithm described as Algorithm \ref{alg:classical_maximal_matching}.

\begin{algorithm}[H]
    \KwInput{An unweighted graph $G$, an integer $k$}
    \KwOutput{A maximal matching $M$}
    $M \gets \emptyset$ \\
    $L \leftarrow $ all $0$ array of size $n$ \\
    \For{$v \in V(G)$} {
        \For{$u \in V(G)$} {
            \lIf{$L(v) = 1$ or $L(u) = 1$} {
                \Continue
            }
            \If{$(v, u) \in E(G)$ confirmed by a query} {
                $M \gets M \cup \{(v, u)\}$ \\
                $L(i) \gets 1$, $L(j) \gets 1$ \\
            }
        }
    }
    \Return {the maximal matching $M$}
    \caption{\texttt{ClassicalMaximalMatching}}\label{alg:classical_maximal_matching}
\end{algorithm}

The depth of the decision tree is $T \leq n^2$.
This algorithm returns a maximal matching of size at most $p$.
Thus, if we always guess that there is no edge between any two vertices, we make at most $I \leq p$ mistakes.
Then, Theorem \ref{thm:guessing_tree_algorithm} implies that the quantum query complexity of the maximal matching problem is $O(\sqrt{TI}) = O(\sqrt{p} n)$, which completes the proof.
\end{proof}

\section{Lower Bounds} \label{sec:lower_bound}



In Section \ref{subsec:tools_lower_bound}, we provide tools to obtain lower bounds on the quantum and randomized query complexities, and
in Sections \ref{subsec:lower_bound_of_vertex_cover} and \ref{subsec:lower_bound_of_matching}, we present our lower bounds for the {\sc $k$-vertex cover} and {\sc $k$-matching} problems in the adjacency matrix model, respectively.


\subsection{Tools for Lower Bounds} \label{subsec:tools_lower_bound}
There are several methods to show lower bounds on the quantum query complexity, such as 
the polynomial method of Beals--Buhrman--Cleve--Mosca--Wolf \cite{beals2001quantum} and
the adversary method of Ambainis \cite{ambainis2000quantum}.
To obtain our lower bounds in this paper, we use the adversary method of Ambainis.

\begin{theorem}[Adversary method {\cite[Theorem 6]{ambainis2000quantum}}] \label{thm:ambainis_quantum_lower_bound}
Let $f \colon \{0, 1\}^N \rightarrow \{0, 1\}$ be a function, and
let $X$ be a set of elements $x \in \{0, 1\}^N$ such that $f(x) = 1$ and $Y$ be a set of elements $y \in \{0, 1\}^N$ such that $f(y) = 0$.
Let $R \subseteq X \times Y$ be a relation.
Let the values $m, m', l_{x, i}, l'_{y, i}$ for $x \in X, y \in Y$ and $i \in [N]$ be such that the following properties hold.
\begin{itemize}
 \item For every $x \in X$, there are at least $m$ different $y \in Y$ such that $(x, y) \in R$.
 \item For every $y \in Y$, there are at least $m'$ different $x \in X$ such that $(x, y) \in R$.
 \item For every $x \in X$ and $i \in [N]$, there are at most $l_{x, i}$ different $y \in Y$ such that $(x, y) \in R$ and $x_i \neq y_i$.
 \item For every $y \in Y$ and $i \in [N]$, there are at most $l'_{y, i}$ different $x \in X$ such that $(x, y) \in R$ and $x_i \neq y_i$.
\end{itemize}
Then, the query complexity of any quantum algorithm computing $f$ with probability at least $2/3$ is $\Omega\left(\sqrt{\frac{mm'}{l_{max}}}\right)$, where $l_{max}$ is the maximum of $l_{x, i}l'_{y, i}$ over all $(x, y) \in R$ and $i \in [N]$ such that $x_i \neq y_i$.
\end{theorem}

For our lower bounds on the randomized query complexity, we use the following theorem of Aaronson \cite{aaronson2004lower}.

\begin{theorem}[{\cite[Theorem 5]{aaronson2004lower}}] \label{thm:aaronson_classical_lower_bound}
Let $f, X, Y, R, m, m', l_{x, i}, l'_{y, i}$ be as in Theorem \ref{thm:ambainis_quantum_lower_bound}.
Then, the query complexity of any randomized algorithm computing $f$ with probability at least $2/3$ is $\Omega(1/v)$,
where $v$ is the maximum of $\min\{l_{x, i}/m, l'_{y, i}/m'\}$ over all $(x, y) \in R$ and $i \in [N]$ such that $x_i \neq y_i$.
\end{theorem}

\subsection{Lower Bounds for Vertex Cover} \label{subsec:lower_bound_of_vertex_cover}



In this subsection, we first present an $\Omega(\sqrt{k} n)$ lower bound on the quantum query complexity of the {\sc $k$-vertex cover} problem when $k \leq n / 3$.
Next, by extending the proof for the case of $k \leq n / 3$ to the case of $k \leq (1 - \epsilon) n$, we obtain an $\Omega(\sqrt{k} n)$ lower bound on the quantum query complexity for any $k \leq (1 - \epsilon) n$.
Finally, we present an $\Omega(n^2)$ lower bound on the randomized query complexity of the {\sc $k$-vertex cover} problem for any $k < n - 1$.

\begin{lemma} \label{thm:lower_bound_of_vertex_cover_small_k}
For any integer $k \leq n / 3$, 
given a graph $G$ with n vertices in the adjacency matrix model,
the quantum query complexity of deciding whether $G$ has a vertex cover of size at most $k$ with bounded error is $\Omega(\sqrt{k} n)$.
\end{lemma}

\begin{proof}
Let $X$ and $Y$ be the sets of graphs on $n$ vertices defined as follows:
\begin{align*}
X = & \left\{(V, M) \, \middle| \, M \subseteq \binom{V}{2} \text{ is a matching with } |M| = k \right\}, \\
Y = & \left\{(V, M) \, \middle| \, M \subseteq \binom{V}{2} \text{ is a matching with } |M| = k + 1 \right\}.
\end{align*}
Clearly, all graphs in $X$ have a vertex cover of size $k$, and all graphs in $Y$ do not have a vertex cover of size at most $k$.
Let $(x, y) \in R$ if $E(x) \subset E(y)$.
Now, we compute the relevant quantities from Theorem \ref{thm:ambainis_quantum_lower_bound}.
Recall that $m$ is the minimum number of graphs $y \in Y$ that each graph $x \in X$ is related to.
Each graph in $X$ can be transformed to a related graph in $Y$ by adding an edge. 
This yields $m = (n - 2 k)(n - 2k - 1) / 2 = \Omega(n^2)$, where we recall that $k \leq n / 3$.
Similarly, we obtain $m' = \Omega(k)$ since
each graph in $Y$ can be transformed to a related graph in $X$ by deleting an edge. 
For any related pairs $(x, y) \in R$,
there is no edge that is present in $x$ and absent in $y$.
Then, $l_{x, i}$ is the maximum number of graphs $y \in Y$ that are related to $x \in X$ where the $i$-th edge is absent in $x$ and present in $y$.
This yields $l_{x, i} \leq 1$, and similarly we can obtain $l'_{y, i} \leq 1$.
Therefore, by Theorem \ref{thm:ambainis_quantum_lower_bound}, we obtain a lower bound of $\Omega(\sqrt{n^2 k}) = \Omega(\sqrt{k}n)$ when $k \leq n / 3$, which completes the proof.
\end{proof}

Next, by extending the proof for the case of $k \leq n / 3$ to the case of $k \leq (1 - \epsilon) n$, we present an $\Omega(\sqrt{k} n)$ lower bound on the quantum query complexity for any $k \leq (1 - \epsilon) n$.\footnote{In the case where $k = n - 2$, we only have to check whether the input graph $G$ is a complete graph. Thus, when $k = n - 2$, we can solve the {\sc $k$-vertex cover} problem using $O(n)$ queries by simply using Grover search. This implies that the $\Omega(\sqrt{k} n)$ lower bound is not true when $k$ is  sufficiently close to $n$.}
In the proof of Lemma \ref{thm:lower_bound_of_vertex_cover_small_k}, we consider graphs that consist of only isolated edges in the construction of the set of graphs $X$ and $Y$.
In this construction, if we try to construct a graph that does not contain a vertex cover of size at most $k$, the graph must have at least $2 (k + 1)$ vertices.
Instead, in the proof for the case of $k \leq (1 - \epsilon) n$, we consider graphs that consist of only cliques of size $t$.
In this construction, if we try to construct a graph that does not contain a vertex cover of size at most $k$, the graph must have at least $\lceil \frac{k + 1}{t - 1} \rceil \cdot t$ vertices.
This is because a clique of size $t$ does not contain a vertex cover of size less than $t - 1$.
Now, we prove Theorem \ref{thm:main_theorem_matrix_model_lower_bound}, which is restated below.


\quantumvclb*

\begin{proof}
By Lemma \ref{thm:lower_bound_of_vertex_cover_small_k}, it remains to prove the claim for the case of $0 < \epsilon < 2/3$.
Now, fix a constant $0 < \epsilon < 2/3$.
Let $\displaystyle t = \biggl\lceil \frac{3}{\epsilon} \biggr\rceil$, and $\displaystyle c = \biggl\lfloor \frac{k}{t - 1} \biggr\rfloor$.
We note that $\displaystyle \biggl\lceil \frac{k + 1}{t - 1} \biggr\rceil = c + 1$.
Here, $c \cdot t \leq \frac{k}{t - 1} \cdot t = k + \frac{k}{t - 1} \leq \left(1 - \epsilon\right) n + \frac{\epsilon n}{2} = \left(1 - \frac{\epsilon}{2}\right) n$.
Then, we have $n - c \cdot t \geq \frac{\epsilon n}{2} = \Omega(n)$, where we recall that $\epsilon = \Theta(1)$.
We also have $(c + 1) \cdot t \leq \left(1 - \frac{\epsilon}{2}\right) n + \lceil \frac{3}{\epsilon} \rceil \leq n$.
Let $X$ and $Y$ be the sets of graphs on $n$ vertices defined as follows:
\begin{align*}
X = & \left\{ \left(V,\, \bigcup_{i = 1}^{c} \binom{W_i}{2} \right) \, \middle| \, \text{$\{W_i\subseteq V\}_{i\in[c]}$ is a family of disjoint subsets of size $t$} \right\}, \\
Y = & \left\{ \left(V,\, \bigcup_{i = 1}^{c + 1} \binom{W_i}{2} \right) \, \middle| \, \text{$\{W_i\subseteq V\}_{i\in[c + 1]}$ is a family of disjoint subsets of size $t$} \right\}.
\end{align*}
Note that all graphs in $X$ have a vertex cover of size at most $k$, and all graphs in $Y$ do not.
Let $(x, y) \in R$ if $E(x) \subset E(y)$.
Now, we compute the relevant quantities from Theorem \ref{thm:ambainis_quantum_lower_bound}.
Recall that $m$ is the minimum number of graphs $y \in Y$ that each graph $x \in X$ is related to.
Each graph in $X$ can be transformed into a related graph in $Y$ by adding a clique of size $t$. 
This yields $m = \binom{n - c \cdot t}{t}$.
Similarly, we obtain $m' = c + 1 = \Omega(k)$ since
each graph in $Y$ can be transformed into a related graph in $X$ by deleting a clique of size $t$. 
For any related pairs $(x, y) \in R$,
there is no edge that is present in $x$ and absent in $y$.
Then, $l_{x, i}$ is the maximum number of graphs $y \in Y$ that are related to $x \in X$ where the $i$--th edge is absent in $x$ and is present in $y$.
This yields $l_{x, i} \leq \binom{n - c \cdot t - 2}{t - 2}$, and similarly we can obtain $l'_{y, i} \leq 1$.
Therefore, by Theorem \ref{thm:ambainis_quantum_lower_bound}, we obtain a lower bound of $\Omega(\sqrt{n^2 k}) = \Omega(\sqrt{k}n)$, which completes the proof.
\end{proof}

By using Theorem \ref{thm:aaronson_classical_lower_bound}, we can also obtain the lower bound on the randomized query complexity of the {\sc $k$-vertex cover} problem.
Now, we prove Proposition \ref{thm:main_theorem_classical_lower_bound}, which is restated below.


\classicalvclb*

\begin{proof}
The proof will be divided into two steps.
We first give the proof in the case where $k \leq n / 3$, and next, we give the proof in the case where $k \geq n / 3$.

Now, we present an $\Omega(n^2)$ lower bound in the case where $k \leq n / 3$.
We apply, for Theorem \ref{thm:aaronson_classical_lower_bound}, $X, Y, R \subseteq X \times Y$ from the proof of Lemma \ref{thm:lower_bound_of_vertex_cover_small_k}.
By the arguments in the proof of Lemma \ref{thm:lower_bound_of_vertex_cover_small_k},
the relevant quantity $v$ from Theorem \ref{thm:aaronson_classical_lower_bound} is $O(1 / n^2)$.
Thus, by Theorem \ref{thm:aaronson_classical_lower_bound}, the lower bound on the randomized query complexity of the {\sc $k$-vertex cover} problem is $\Omega(n^2)$ when $k \leq n / 3$.

Next we present an $\Omega(n^2)$ lower bound in the case where $k \geq n / 3$.
Let $X$ and $Y$ be the sets of graphs on $n$ vertices defined as follows:
\begin{align*}
X = & \left\{\left(V, \binom{W}{2} \setminus \{ (v, u)\}\right) \,\middle|\, W \subseteq V,\,  |W|= k+2, \text{$v \in W$, and $u \in W$} \right\}, \\
Y = & \left\{\left(V, \binom{W}{2} \right) \,\middle|\, W \subseteq V,\,  |W|= k+2 \right\}.
\end{align*}
Note that all graphs in $X$ have a vertex cover of size $k$, and all graphs in $Y$ do not have a vertex cover of size at most $k$.
Let $(x, y) \in R$ if $E(x) \subset E(y)$.
Now, we compute the relevant quantities from Theorem \ref{thm:ambainis_quantum_lower_bound}.
Recall that $m'$ is the minimum number of graphs $x \in X$ that each graph $y \in Y$ is related to.
Each graph in $Y$ can be transformed into a related graph in $X$ by deleting an edge. 
This yields $m' = (k + 2)(k + 1) / 2 = \Omega(n^2)$, where we recall that $k \geq n / 3$.
Similarly, we obtain $m = 1$ since
each graph in $X$ can be transformed to a related graph in $Y$ by adding the edge $(v, u)$. 
For any related pairs $(x, y) \in R$,
there is no edge that is present in $x$ and absent in $y$.
Then, $l_{x, i}$ is the maximum number of graphs $y \in Y$ that are related to $x \in X$ where the $i$--th edge is absent in $x$ and present in $y$.
This yields $l_{x, i} \leq 1$, and similarly we can obtain $l'_{y, i} \leq 1$.
Here, the relevant quantity $v$ from Theorem \ref{thm:aaronson_classical_lower_bound} is $O(1 / n^2)$.
Thus, by Theorem \ref{thm:aaronson_classical_lower_bound}, the lower bound on the randomized query complexity for the {\sc $k$-vertex cover} problem is $\Omega(n^2)$ in the case where $k \geq n / 3$, which completes the proof.
\end{proof}

\subsection{Lower Bounds for Matching} \label{subsec:lower_bound_of_matching}

In this subsection, we first present an $\Omega(\sqrt{k} n)$ lower bound on the quantum query complexity of the {\sc $k$-matching} problem.
Next, we present an $\Omega(n^2)$ lower bound on the randomized query complexity of the {\sc $k$-matching} problem.


First, we prove Theorem \ref{thm:main_theorem_matrix_model_lower_bound_matching}, which is restated below.

\quantummatchinglb*

\begin{proof}
The proof will be divided into two steps.
We first give the proof of an $\Omega(\sqrt{k}n)$ lower bound for the case of $k \leq n / 3$, which is almost the same as the proof of Lemma \ref{thm:lower_bound_of_vertex_cover_small_k}.
Next we give the proof of an $\Omega(k^{3/2})$ lower bound for any $k \leq n / 2$,
which is almost the same as the proof of the $\Omega(n^{3/2})$ lower bound for the bipartite perfect matching problem by Zhang \cite{zhang2004power}.

Now, we present an $\Omega(\sqrt{k}n)$ lower bound in the case of $k \leq n / 3$.
We define $X, Y, R \subseteq X \times Y$ in almost the same way as in the proof of Lemma \ref{thm:lower_bound_of_vertex_cover_small_k}.
Let $X$ and $Y$ be the sets of graphs on $n$ vertices defined as follows:
\begin{align*}
X = & \left\{(V, M) \, \middle| \, M \subseteq \binom{V}{2} \text{ is a matching with } |M| = k - 1 \right\}, \\
Y = & \left\{(V, M) \, \middle| \, M \subseteq \binom{V}{2} \text{ is a matching with } |M| = k \right\}.
\end{align*}
Note that all graphs in $X$ do not have a matching of size at least $k$, and all graphs in $Y$ have a matching of size $k$.
Let $(x, y) \in R$ if $E(x) \subset E(y)$.
By the same argument as the proof of Theorem \ref{thm:lower_bound_of_vertex_cover_small_k}, we obtain an $\Omega(\sqrt{k}n)$ lower bound when $k \leq n / 3$.

Next, we show an $\Omega(k^{3/2})$ lower bound.
We define $X', Y', R' \subseteq X' \times Y'$ in almost the same way as in the proof of the $\Omega(n^{3/2})$ lower bound for the bipartite perfect matching problem by Zhang \cite{zhang2004power}.
Let $X'$ and $Y'$ be the sets of graphs on $n$ vertices defined as follows:
\begin{equation*}
X' = \left\{(V, C) \, \middle| \, C \subseteq \binom{V}{2} \text{ is a cycle such that  } |C| = 2 k \right\}, 
\end{equation*}

\begin{equation*}
Y' = \left\{(V, C_1 \cup C_2) \, \middle| \, 
\begin{gathered}
 \text{$C_1, C_2 \subseteq \binom{V}{2}$ are two disjoint cycles such that } |C_1| + |C_2| = 2 k \text{,} \\ |C_1| \geq \frac{k}{2} \text{, } |C_2| \geq \frac{k}{2} \text{, and $|C_1|$ is odd}
\end{gathered}
 \right\}.
\end{equation*}
Note that all graphs in $X'$ have a matching of size $k$, and all graphs in $Y'$ do not have a matching of size at least $k$.
Let $(x', y') \in R'$ if there exists 4 vertices $a, b, c, d \in V$ such that $(a, b), (c, d) \in E(x')$, and $E(x') \setminus \{(a, b), (c, d)\} \cup \{(a, c),(b, d)\} = E(y')$. 
By the same argument of the proof of the $\Omega(n^{3/2})$ lower bound for the bipartite perfect matching problem by Zhang \cite{zhang2004power}, we obtain an $\Omega(k^{3/2})$ lower bound.

Therefore, we obtain an $\Omega(\sqrt{k}n)$ lower bound for any $k \leq n / 2$, which completes the proof.
\end{proof}


Next, we prove Proposition \ref{thm:main_theorem_matrix_model_classical_lower_bound_matching}, which is restated below.

\classicalmatchinglb*

\begin{proof}
We apply, for Theorem \ref{thm:aaronson_classical_lower_bound}, $X, Y, R \subseteq X \times Y$ and $X', Y', R' \subseteq X' \times Y'$ from the proof of Theorem \ref{thm:main_theorem_matrix_model_lower_bound_matching}.
By the arguments in the proof of lemma \ref{thm:lower_bound_of_vertex_cover_small_k},
the relevant quantity $v$ from Theorem \ref{thm:aaronson_classical_lower_bound} is $O(1 / n^2)$.
Thus, by Theorem \ref{thm:aaronson_classical_lower_bound}, the lower bound on the randomized query complexity of the {\sc $k$-matching} problem is $\Omega(n^2)$.
\end{proof}

\section{Concluding Remarks} \label{sec:conclusion}

In this paper, we consider the parameterized quantum query complexities of the {\sc $k$-vertex cover} and {\sc $k$-matching} problems.
We obtain the optimal query complexities for the {\sc $k$-vertex cover} and {\sc $k$-matching} problems for small $k$.

For the {\sc $k$-vertex cover} problem, our key observation in this paper is that the method of kernelization is useful for designing quantum query algorithms for parameterized problems.
We emphasize that this paper is the first work to demonstrate the usefulness of kernelization for the parameterized quantum query complexity.
On the other hand, there is a limitation of the techniques based on the method of kernelization.
We defined a quantum query kernelization algorithm as an algorithm that converts an input graph given as a quantum oracle into a kernel as a bit string.
Here, we observe that on the basis of our quantum query kernelization approach, it is impossible to show that the quantum query complexity of the {\sc $k$-vertex cover} problem is $k^{2 - \Omega(1)}$
since there is no kernel $(G', k')$ consisting of $O(k^{2-\epsilon})$ edges unless the polynomial hierarchy collapses \cite{dell2014satisfiability}.

It would be interesting to consider quantum-to-quantum kernelization algorithms, which convert an instance given as a quantum oracle into a smaller instance expressed by a quantum oracle.

For the {\sc $k$-matching} problem, we designed a quantum query algorithm, which iteratively finds an augmenting path or a vertex with at most one neighbor outside of the current maximal matching.
While our $k$-matching algorithm is optimal for small $k$, its query complexity is $\Theta(n^2)$ for $k = \Omega(n)$.
On the other hand, the currently known best non-parameterized upper bound is $O(n^{7/4})$ \cite{kimmel2021query}. 
We would be able to improve the parameterized quantum query complexity for the {\sc $k$-matching} problem for large $k$ by combining our approach with existing or new approaches.

\bibliography{bib2doi}

\end{document}